%% file: main.tex
\documentclass[11pt]{article}

\include{macros}

\title{Adversarially Robust Coloring for Graph Streams
  \thanks{This work was supported in part by NSF under awards 1907738 and 2006589.}
  }
\author{Amit Chakrabarti
  \thanks{Department of Computer Science, Dartmouth College.}
  \and
  Prantar Ghosh $^\fnsymbol{footnote}$
  \and
  Manuel Stoeckl $^\fnsymbol{footnote}$
}

\date{}

\begin{document}

\maketitle
\thispagestyle{empty}
\input{abstract}

\vskip5ex
\centerline{\textbf{Keywords:~} data streaming;~ graph algorithms;~ graph coloring;~ lower bounds;~ online algorithms}

\clearpage
\addtocounter{page}{-1}

\input{intro}

\input{techniques}

\input{prelims}

\input{lb}
\input{ubs}

\section*{Acknowledgements}

Prantar Ghosh would like to thank Sayan Bhattacharya for a helpful conversation regarding this work. 

\bibliographystyle{alpha}
\bibliography{refs}

\end{document}

%% file: macros.tex
\usepackage[margin=1in]{geometry}
\usepackage{amsmath, amsthm, amssymb, thmtools, mleftright}
\usepackage{xfrac} 
\usepackage{xspace, paralist, enumitem, multirow, booktabs, graphicx, soul}
\usepackage[usenames, dvipsnames]{xcolor}
\usepackage{times, mathptmx, bm}
\usepackage[utf8]{inputenc}
\usepackage[T1]{fontenc}
\usepackage{algorithm}
\usepackage[noend]{algpseudocode}
\usepackage{url}
\usepackage[colorlinks, linkcolor=BrickRed, citecolor=blue]{hyperref}
\usepackage{cite, cleveref}

\newcommand*\Let[2]{\State #1 $\gets$ #2}

\allowdisplaybreaks[1]

\newtheorem{theorem}{Theorem}[section]
\newtheorem{lemma}[theorem]{Lemma}
\newtheorem{corollary}[theorem]{Corollary}
\newtheorem{fact}{Fact}[section]

\theoremstyle{definition}  \newtheorem{definition}[theorem]{Definition}

\theoremstyle{remark}  

\DeclareMathOperator{\poly}{poly}
\DeclareMathOperator{\polylog}{polylog}
\DeclareMathOperator{\msg}{msg}
\DeclareMathOperator{\out}{out}
\DeclareMathOperator{\cost}{cost}
\DeclareMathOperator{\R}{R}
\DeclareMathOperator{\D}{D}
\DeclareMathOperator{\image}{Im}

\newcommand{\EE}{\mathbb{E}}

\newcommand{\NN}{\mathbb{N}}

\newcommand{\cA}{\mathcal{A}}

\newcommand{\cC}{\mathcal{C}}
\newcommand{\cD}{\mathcal{D}}

\newcommand{\cO}{\mathcal{O}}

\newcommand{\cQ}{\mathcal{Q}}
\newcommand{\cR}{\mathcal{R}}

\newcommand{\cU}{\mathcal{U}}

\newcommand{\cX}{\mathcal{X}}

\newcommand{\cZ}{\mathcal{Z}}

\newcommand{\tO}{\widetilde{O}}
\newcommand{\tOmega}{\widetilde{\Omega}}

\newcommand{\eps}{\varepsilon}
\newcommand{\mypar}[1]{\medskip\noindent{\bfseries #1.}~}

\newcommand{\avoid}{\textsc{avoid}\xspace}
\newcommand{\kavoid}{\ensuremath{\textsc{avoid}^k}\xspace}

\newcommand{\mif}{\textsc{mif}\xspace}
\newcommand{\degr}{\textsc{deg}\xspace}
\newcommand{\clr}{\textsc{clr}\xspace}
\newcommand{\chunksize}{\textsc{ChunkSize}\xspace}
\newcommand{\checkpointmaxdeg}{\textsc{CheckptMaxDeg}\xspace}
\newcommand{\used}{\textsc{used}\xspace}
\newcommand{\init}{\textsc{init}\xspace}
\newcommand{\process}{\textsc{process}\xspace}
\newcommand{\query}{\textsc{query}\xspace}
\newcommand{\next}{\textsc{next}\xspace}
\newcommand{\insedge}{\textsc{ins-edge}\xspace}
\newcommand{\deledge}{\textsc{del-edge}\xspace}

\newcommand{\ceil}[1]{{\left\lceil{#1}\right\rceil}}
\newcommand{\floor}[1]{{\left\lfloor{#1}\right\rfloor}}
\newcommand{\setm}{\smallsetminus}
\renewcommand{\b}{\{0,1\}}

%% file: abstract.tex

\begin{abstract}
  A streaming algorithm is considered to be adversarially robust if it provides correct outputs with high probability even when the stream updates are chosen by an adversary who may observe and react to the past outputs of the algorithm. We grow the burgeoning body of work on such algorithms in a new direction by studying robust algorithms for the problem of maintaining a valid vertex coloring of an $n$-vertex graph given as a stream of edges. Following standard practice, we focus on graphs with maximum degree at most $\Delta$ and aim for colorings using a small number $f(\Delta)$ of colors. 

  A recent breakthrough (Assadi, Chen, and Khanna; SODA~2019) shows that in the standard, non-robust, streaming setting, $(\Delta+1)$-colorings can be obtained while using only $\widetilde{O}(n)$ space. Here, we prove that an adversarially robust algorithm running under a similar space bound must spend almost $\Omega(\Delta^2)$ colors and that robust $O(\Delta)$-coloring requires a \emph{linear} amount of space, namely $\Omega(n\Delta)$. We in fact obtain a more general lower bound, trading off the space usage against the number of colors used. From a complexity-theoretic standpoint, these lower bounds provide (i)~the first significant separation between adversarially robust algorithms and ordinary randomized algorithms for a \emph{natural} problem on insertion-only streams and (ii)~the first significant separation between randomized and deterministic coloring algorithms for graph streams, since deterministic streaming algorithms are automatically robust.

  We complement our lower bounds with a suite of positive results, giving adversarially robust coloring algorithms using sublinear space. In particular, we can maintain an $O(\Delta^2)$-coloring using $\widetilde{O}(n \sqrt{\Delta})$ space and an $O(\Delta^3)$-coloring using $\widetilde{O}(n)$ space.
\end{abstract}

%% file: intro.tex

\section{Introduction} \label{sec:intro}

A data streaming algorithm processes a huge input, supplied as a long sequence
of elements, while using working memory (i.e., space) much smaller than the
input size. The main algorithmic goal is to compute or estimate some function
of the input $\sigma$ while using space \emph{sublinear} in the size of
$\sigma$. For most---though not all---problems of interest, a streaming
algorithm \emph{needs} to be randomized in order to achieve sublinear space.
For a randomized algorithm, the standard correctness requirement is that for
each possible input stream it return a valid answer with high probability. A
burgeoning body of work---much of it very recent~\cite{BenEliezerJWY20,
BenEliezerY20, HassidimKMMS20, KaplanMNS21, BravermanHMSSZ21, WoodruffZ21,
AttiasCSS21, beneliezerEO21} but traceable back to~\cite{HardtW13}---addresses streaming
algorithms that seek an even stronger correctness guarantee, namely that they
produce valid answers with high probability even when working with an input
generated by an active adversary.  There is compelling motivation from
practical applications for seeking this stronger guarantee: for instance,
consider a user continuously interacting with a database and choosing future
queries based on past answers received; or think of an online streaming or
marketing service looking at a customer’s transaction history and recommending
them products based on it.

We may view the operation of streaming algorithm $\cA$ as a game between a
\emph{solver}, who executes $\cA$, and an \emph{adversary}, who generates a
``hard'' input stream $\sigma$. The standard notion of $\cA$ having error
probability $\delta$ is that for every fixed $\sigma$ that the adversary may
choose, the probability over $\cA$'s random choices that it errs on $\sigma$
is at most $\delta$. Since the adversary has to make their choice before the
solver does any work, they are {\em oblivious} to the actual actions of the
solver. In contrast to this, an {\em adaptive adversary} is not required to
fix all of $\sigma$ in advance, but can generate the elements (tokens) of
$\sigma$ incrementally, based on outputs generated by the solver as it
executes $\cA$. Clearly, such an adversary is much more powerful and can
attempt to learn something about the solver's internal state in order to
generate input tokens that are bad for the particular random choices made by
$\cA$. Indeed, such adversarial attacks are known to break many well known
algorithms in the streaming literature~\cite{HardtW13, BenEliezerJWY20}.
Motivated by this, one defines a $\delta$-error \emph{adversarially robust
streaming algorithm} to be one where the probability that an adaptive
adversary can cause the solver to produce an incorrect output at some point of
time is at most $\delta$. Notice that a \emph{deterministic} streaming
algorithm (which, by definition, must always produce correct answers) is
automatically adversarially robust. 

Past work on such adversarially robust streaming algorithms has focused on
statistical estimation problems and on sampling problems but, with the
exception of~\cite{BravermanHMSSZ21}, there has not been much study of graph
theoretic problems. This work focuses on graph coloring, a fundamental
algorithmic problem on graphs. Recall that the goal is to efficiently process
an input graph given as a stream of edges and assign colors to its vertices
from a small palette so that no two adjacent vertices receive the same color.
The main messages of this work are that (i)~while there exist surprisingly
efficient sublinear-space algorithms for coloring under standard streaming, it
is provably harder to obtain adversarially robust solutions; but nevertheless,
(ii)~there {\em do} exist nontrivial sublinear-space robust algorithms for
coloring.

To be slightly more detailed, suppose we must color an $n$-vertex input graph
$G$ that has maximum degree~$\Delta$. Producing a coloring using only
$\chi(G)$ colors, where $\chi(G)$ is the chromatic number, is \textsf{NP}-hard
while producing a $(\Delta+1)$-coloring admits a straightforward greedy
algorithm, given offline access to $G$. Producing a good coloring given only
streaming access to $G$ and sublinear (i.e., $o(n\Delta)$ bits of) space is a
nontrivial problem and the subject of much recent research~\cite{BeraG18,
AssadiCK19, BeraCG20, AlonA20, BhattacharyaBMU21}, including the breakthrough
result of Assadi, Chen, and Khanna \cite{AssadiCK19} that gives a
$(\Delta+1)$-coloring algorithm using only semi-streaming (i.e., $\tO(n)$ bits
of) space.\footnote{The notation $\tO(\cdot)$ hides factors polylogarithmic in
$n$.} However, all of these algorithms were designed with only the standard,
oblivious adversary setting in mind; an adaptive adversary can make all of
them fail. This is the starting point for our exploration in this work. 

\subsection{Our Results and Contributions} \label{sec:results}

We ask whether the graph coloring problem is inherently harder under an
adversarial robustness requirement than it is for standard streaming.  We
answer this question affirmatively with the first major theorem in this work,
which is the following (we restate the theorem with more detail and formality
as \Cref{thm:lower-bound-core}).
\begin{theorem}\label{thm:lb:preview}
  A constant-error adversarially robust algorithm that processes a stream of edge 
  insertions into an $n$-vertex graph and, as long as the maximum degree of the 
  graph remains at most $\Delta$, maintains a valid $K$-coloring (with
  $\Delta+1 \le K \le n/2$) must use at least $\Omega(n\Delta^2/K)$ bits of space.
\end{theorem}

\noindent We spell out some immediate corollaries of this result because of
their importance as conceptual messages.

\begin{itemize}[leftmargin=\parindent,topsep=4pt,itemsep=1pt]
  \item \textbf{Robust coloring using $O(\Delta)$ colors.~} In the setting of
  \Cref{thm:lb:preview}, if the algorithm is to use only $O(\Delta)$ colors,
  then it must use $\Omega(n\Delta)$ space. In other words, a sublinear-space
  solution is ruled out. 
  \item \textbf{Robust coloring using semi-streaming space.~} In the setting
  of \Cref{thm:lb:preview}, if the algorithm is to run in only $\tO(n)$ space,
  then it must use $\tOmega(\Delta^2)$ colors.
  \item \textbf{Separating robust from standard streaming with a natural
  problem.~} Contrast the above two lower bounds with the guarantees of the
  \cite{AssadiCK19} algorithm, which handles the non-robust case. This shows
  that ``maintaining an $O(\Delta)$-coloring of a graph'' is a natural (and
  well-studied) algorithmic problem where, even for insertion-only streams,
  the space complexities of the robust and standard streaming versions of the
  problem are well separated: in fact, the separation is roughly quadratic, by
  taking $\Delta = \Theta(n)$. This answers an open question
  of~\cite{KaplanMNS21}, as we explain in greater detail in
  \Cref{sec:related}.
  \item \textbf{Deterministic versus randomized coloring.~} Since every
  deterministic streaming algorithm is automatically adversarially robust, the
  lower bound in \Cref{thm:lb:preview} applies to such algorithms. In
  particular, this settles the deterministic complexity of
  $O(\Delta)$-coloring. Also, turning to semi-streaming algorithms, whereas a
  combinatorially optimal\footnote{If one must use at most $f(\Delta)$ colors
  for some function $f$, the best possible function that always works is
  $f(\Delta) = \Delta+1$.} $(\Delta+1)$-coloring is possible using
  randomization~\cite{AssadiCK19}, a deterministic solution must spend at
  least $\tOmega(\Delta^2)$ colors. These results address a broadly-stated
  open question of Assadi~\cite{Assadi-detcolor-talk}; see \Cref{sec:related}
  for details.
\end{itemize}

We prove the lower bound in \Cref{thm:lb:preview} using a reduction from a
novel two-player communication game that we call \textsc{subset-avoidance}. In
this game, Alice is given an $a$-sized subset of the universe~$[t]$;\footnote{The notation $[t]$ denotes the set $\{1,2,\ldots,t\}$.} she
must communicate a possibly random message to Bob that causes him to output a
$b$-sized subset of~$[t]$ that, with high probability, avoids Alice's set
completely. We give a fairly tight analysis of the communication complexity of
this game, showing an $\Omega(ab/t)$ lower bound, which is matched by an
$\tO(ab/t)$ deterministic upper bound.  The \textsc{subset-avoidance} problem is
a natural one.  We consider the definition of this game and its
analysis---which is not complicated---to be additional conceptual
contributions of this work; these might be of independent interest for future
applications. 

We complement our lower bound with some good news: we give a suite of
upper bound results by designing adversarially robust coloring algorithms that
handle several interesting parameter regimes. Our focus is on maintaining a
valid coloring of the graph using $\poly(\Delta)$ colors, where $\Delta$ is
the current maximum degree, as an adversary inserts edges. In fact, some of
these results hold even in a \emph{turnstile model}, where the adversary
might both add and delete edges.
In this context, it is worth noting that the \cite{AssadiCK19} algorithm also works in a turnstile setting.
\begin{theorem} \label{thm:ub:preview}
  There exist adversarially robust algorithms for coloring an $n$-vertex graph
  achieving the following tradeoffs (shown in \Cref{table:ub}) between the
  space used for processing the stream and the number of colors spent, where
  $\Delta$ denotes the evolving maximum degree of the graph and, in the
  turnstile setting, $m$ denotes a known upper bound on the stream length.
  \begin{table*}[!hbt] \centering
  \begin{tabular}{c c c c c}
    \toprule
      {\bf Model}
      & {\bf Colors}  
      & {\bf Space} 
      & {\bf Notes}
      & {\bf Reference}\\
    \midrule
      Insertion-only & $O(\Delta^3)$ & $\tO(n)$ & $\tO(n\Delta)$ external random bits & \Cref{thm:ub-delta3} \\
      Insertion-only & $O(\Delta^k)$ & $\tO(n\Delta^{1/k})$ & any $k\in \NN$ & \Cref{cor:insertonly-ub-general}\\
      Strict Graph Turnstile & $O(\Delta^k)$ & $\tO(n^{1-1/k}m^{1/k})$ & constant $k\in \NN$ & \Cref{thm:turnstile-ub-general}\\
    \bottomrule
  \end{tabular}
  \caption{A summary of our adversarially robust coloring algorithms. A ``strict graph turnstile'' model requires the input to describe a simple graph at all times; see \Cref{sec:prelims}.}
  \label{table:ub}
  \end{table*}

  In each of these algorithms, for each stream update or query made by the
  adversary, the probability that the algorithm fails either by returning an
  invalid coloring or aborting is at most $1/\poly(n)$.
\end{theorem}

We give a more detailed discussion of these results, including an explanation
of the technical caveat noted in \Cref{table:ub} for the
$O(\Delta^3)$-coloring algorithm, in \Cref{sec:techniques-ub}. 

\subsection{Motivation, Context, and Related Work} 
\label{sec:motivation} \label{sec:related}

Graph streaming has become widely popular~\cite{McGregor14}, especially since
the advent of large and evolving networks including social media, web graphs,
and transaction networks. These large graphs are regularly mined for knowledge
and such knowledge often informs their future evolution. Therefore, it is
important to have adversarially robust algorithms for working with these
graphs. Yet, the recent explosion of interest in robust algorithms has not
focused much on graph problems. We now quickly recap some history.

Two influential works~\cite{MironovNS11, HardtW13} identified the challenge
posed by adaptive adversaries to sketching and streaming algorithms. In
particular, Hardt and Woodruff~\cite{HardtW13} showed that many statistical
problems, including the ubiquitous one of $\ell_2$-norm estimation, do not
admit adversarially robust linear sketches of sublinear size. Recent works
have given a number of positive results.  Ben-Eliezer, Jayaram, Woodruff, and
Yogev~\cite{BenEliezerJWY20} considered such fundamental problems as distinct
elements, frequency moments, and heavy hitters (these date back to the
beginnings of the literature on streaming algorithms); for
$(1\pm\eps)$-approximating a function value, they gave two generic frameworks
that can ``robustify'' a standard streaming algorithm, blowing up the space
cost by roughly the \emph{flip number} $\lambda_{\eps,m}$, defined as the
maximum number of times the function value can change by a factor of
$1\pm\eps$ over the course of an $m$-length stream.  For \emph{insertion-only
streams} and monotone functions, $\lambda_{\eps,m}$ is roughly
$O(\eps^{-1}\log m)$, so this overhead is very small. Subsequent works
\cite{HassidimKMMS20, WoodruffZ21, AttiasCSS21} have improved this overhead
with the current best-known one being
$O\left(\sqrt{\eps\lambda_{\eps,m}}\right)$ \cite{AttiasCSS21}. 

For insertion-only graph streams, a number of well-studied problems such as
triangle counting, maximum matching size, and maximum subgraph density can be
handled by the above framework because the underlying functions are monotone.
For some problems such as counting connected components, there are simple
deterministic algorithms that achieve an asymptotically optimal space bound,
so there is nothing new to say in the robust setting. For graph
sparsification, \cite{BravermanHMSSZ21} showed that the Ahn--Guha sketch
\cite{AhnG09} can be made adversarially robust with a slight loss in the
quality of the sparsifier. Thanks to efficient adversarially robust sampling
\cite{BenEliezerY20, BravermanHMSSZ21}, many sampling-based graph algorithms
should yield corresponding robust solutions without much overhead. For problems 
calling for Boolean answers, such as testing connectivity or bipartiteness,
achieving low error against an oblivious adversary automatically does so against
an adaptive adversary as well, since a sequence of correct outputs from the
algorithm gives away no information to the adversary. This is a particular case
of a more general phenomenon captured by the notion of pseudo-determinism, 
discussed at the end of this section.

Might it be that for all interesting data streaming problems, efficient
standard streaming algorithms imply efficient robust ones? The above framework
does not automatically give good results for \emph{turnstile} streams, where
each token specifies either an insertion or a deletion of an item, or for
estimating non-monotone functions. In either of these situations, the flip
number can be very large. As noted above, linear sketching, which is the
preeminent technique behind turnstile streaming algorithms (including ones for
graph problems), is vulnerable to adversarial attacks~\cite{HardtW13}. This
does not quite provide a separation between standard and robust space
complexities, since it does not preclude efficient non-linear solutions. The
very recent work \cite{KaplanMNS21} gives such a separation: it exhibits a
function estimation problem for which the ratio between the adversarial and
standard streaming complexities is as large as
$\tOmega\left(\sqrt{\lambda_{\eps,m}}\right)$, which is exponential upon
setting parameters appropriately.  However, their function is highly
artificial, raising the important question: \emph{Can a significant gap be
shown for a natural streaming problem?} \footnote{This open question was
explicitly raised in the STOC 2021 workshop \emph{Robust Streaming, Sketching,
and Sampling}~\cite{Stemmer21-stoc}.}

It is easy to demonstrate such a gap in graph streaming.  Consider the problem
of finding a spanning forest in a graph undergoing edge insertions and
deletions. The celebrated Ahn--Guha--McGregor sketch~\cite{AhnGM12} solves
this in $\tO(n)$ space, but this sketch is not adversarially robust. Moreover,
suppose that $\cA$ is an adversarially robust algorithm for this problem.
Then we can argue that the memory state of $\cA$ upon processing an unknown
graph $G$ must contain enough information to recover $G$ entirely: an
adversary can repeatedly ask $\cA$ for a spanning forest, delete all returned
edges, and recurse until the evolving graph becomes empty. Thus, for basic
information theoretic reasons, $\cA$ must use $\Omega(n^2)$ bits of space,
resulting in a quadratic gap between robust and standard streaming space
complexities. Arguably, this separation is not very satisfactory, since the
hardness arises from the turnstile nature of the stream, allowing the
adversary to delete edges. Meanwhile, the~\cite{KaplanMNS21} separation 
does hold for insert-only streams, but as we (and they) note, their problem is
rather artificial.

\mypar{Hardness for Natural Problems}
We now make a simple, yet crucial, observation. Let
\textsc{missing-item-finding} (\mif) denote the problem where, given an
evolving set $S \subseteq [n]$, we must be prepared to return an element in
$[n] \setm S$ or report that none exists. When the elements of $S$ are given
as an input stream, \mif admits the following $O(\log^2 n)$-space solution
against an oblivious adversary: maintain an $\ell_0$-sampling
sketch~\cite{JowhariST11} for the characteristic vector of $[n] \setm S$ and
use it to randomly sample a valid answer. In fact, this solution extends to
turnstile streams.  Now suppose that we have an adversarially robust algorithm
$\cA$ for \mif, handling insert-only streams.  Then, given the memory state of
$\cA$ after processing an unknown set $T$ with $|T| = n/2$, an adaptive
adversary can repeatedly query $\cA$ for a missing item $x$, record $x$,
insert $x$ as the next stream token, and continue until $\cA$ fails to find an
item. At that point, the adversary will have recorded (w.h.p.) the set $[n]
\setm T$, so he can reconstruct $T$. As before, by basic information theory,
this reconstructability implies that $\cA$ uses $\Omega(n)$ space.

This exponential gap between standard and robust streaming, based on well-known results,
seems to have been overlooked---perhaps because \mif does not
conform to the type of problems, namely estimation of real-valued functions,
that much of the robust streaming literature has focused on. That said, though
\mif is a natural problem and the hardness holds for insert-only streams,
there is one important box that \mif does not tick: it is not important enough
on its own and so does not command a serious literature. This leads us to
refine the open question of \cite{KaplanMNS21} thus: \emph{Can a significant
gap be shown for a natural and well-studied problem with the hardness holding
even for insertion-only streams?}

With this in mind, we return to graph problems, searching for such a gap. In
view of the generic framework of \cite{BenEliezerJWY20} and follow-up works,
we should look beyond estimating some monotone function of the graph with
scalar output. What about problems where the output is a big \emph{vector},
such as approximate maximum matching (not just its size) or approximate
densest subgraph (not just the density)? It turns out that the \emph{sketch
switching} technique of \cite{BenEliezerJWY20} can still be applied: since we
need to change the output only when the estimates of the associated numerical
values (matching size and density, respectively) change enough, we can proceed
as in that work, switching to a new sketch with fresh randomness that remains
unrevealed to the adversary. This gives us a robust algorithm incurring only
logarithmic overhead. 

But graph coloring is different. As our \Cref{thm:lb:preview} shows, it does
exhibit a quadratic gap for the right setting of parameters and it is, without
doubt, a heavily-studied problem, even in the data streaming setting. 

The above hardness of
\mif provides a key insight into why graph coloring is hard; see \Cref{sec:techniques-lb}.

\mypar{Connections with Other Work on Streaming Graph Coloring}
Graph coloring is, of course, a heavily-studied problem in theoretical
computer science. For this discussion, we stick to streaming algorithms for
this problem, which already has a significant literature~\cite{BeraG18,
AbboudCKP19, AssadiCK19, BeraCG20, AlonA20, BhattacharyaBMU21}.

Although it is not possible to $\chi(G)$-color an input graph in sublinear
space \cite{AbboudCKP19}, as~\cite{AssadiCK19} shows, there is a
semi-streaming algorithm that produces a $(\Delta+1)$-coloring.  This follows
from their elegant \emph{palette sparsification} theorem, which states that if
each vertex samples roughly $O(\log n)$ colors from a palette of size
$\Delta+1$, then there exists a proper coloring of the graph where each vertex
uses a color only from its sampled list. Hence, we only need to store edges
between vertices whose lists intersect. If the edges of $G$ are independent of
the algorithm's randomness, then the expected number of such ``conflict''
edges is $O(n\log^2 n)$, leading to a semi-streaming algorithm.  But note that
an adaptive adversary can attack this algorithm by using a reported coloring
to learn which future edges would definitely be conflict edges and inserting
such edges to blow up the algorithm's storage.

There are some other semi-streaming algorithms (in the standard setting) that
aim for $\Delta(1+\eps)$-colorings. One is palette-sparsification based
\cite{AlonA20} and so, suffers from the above vulnerability against an
adaptive adversary. Others \cite{BeraG18, BeraCG20} are based on randomly
partitioning the vertices into clusters and storing only intra-cluster edges,
using pairwise disjoint palettes for the clusters.
Here, the
semi-streaming space bound hinges on the random partition being likely to
assign each edge's endpoints to different clusters. This can be broken by an
adaptive adversary, who can use a reported coloring to learn many vertex pairs
that are intra-cluster and then insert new edges at such pairs.

Finally, we highlight an important theoretical question about sublinear
algorithms for graph coloring: \emph{Can they be made deterministic?} This was
explicitly raised by Assadi~\cite{Assadi-detcolor-talk} and, prior to this work,
it was open whether, for $(\Delta+1)$-coloring, \emph{any} sublinear space
bound could be obtained deterministically. Our \Cref{thm:lb:preview} settles
the deterministic space complexity of this problem, showing that even the
weaker requirement of $O(\Delta)$-coloring forces $\Omega(n\Delta)$ space,
which is linear in the input size.

Parameterizing \Cref{thm:lb:preview} differently, we see that a robust (in
particular, a deterministic) algorithm that is limited to semi-streaming space
must spend $\tOmega(\Delta^2)$ colors. A major remaining open question is
whether this can be matched, perhaps by a deterministic semi-streaming
$O(\Delta^2)$-coloring algorithm. In fact, it is not known how to get even a
$\poly(\Delta)$-coloring deterministically.  Our algorithmic results,
summarized in \Cref{thm:ub:preview}, make partial progress on this question.
Though we do not obtain deterministic algorithms, we obtain adversarially
robust ones, and we do obtain $\poly(\Delta)$-colorings, though not all the way
down to $O(\Delta^2)$ in semi-streaming space.

\mypar{Other Related Work} Pseudo-deterministic streaming
algorithms\cite{GoldwasserGMW20} fall between adversarially robust and
deterministic ones. Such an algorithm is allowed randomness, but for each
particular input stream it must produce one fixed output (or output sequence)
with high probability. Adversarial robustness is automatic, because when such
an algorithm succeeds, it does not reveal any of its random bits through the
outputs it gives. Thus, there is nothing for an adversary to base adaptive
decisions on.

The well-trodden subject of dynamic graph algorithms deals with a model
closely related to the adaptive adversary model: one receives a stream of edge
insertions/deletions and seeks to maintain a solution after each update. There
have been a few works on the $\Delta$-based graph coloring problem in this
setting \cite{BhattacharyaCHN18, BhattacharyaGKLS19, HenzingerP20}. However,
the focus of the dynamic setting is on optimizing the update \emph{time}
without any restriction on the space usage; this is somewhat orthogonal to the
streaming setting where the primary goal is space efficiency, and update time,
while practically important, is not factored into the complexity.

%% file: techniques.tex

\section{Overview of Techniques}

\subsection{Lower Bound Techniques} \label{sec:techniques-lb}

As might be expected, our lower bounds are best formalized through communication complexity. Recall that a typical communication-to-streaming reduction for proving a one-pass streaming space lower bound works as follows. We set up a communication game for Alice and Bob to solve, using one message from Alice to Bob. Suppose that Alice and Bob have inputs $x$ and $y$ in this game. The players simulate a purported efficient streaming algorithm $\cA$ (for $P$, the problem of interest) by having Alice feed some tokens into $\cA$ based on $x$, communicating the resulting memory state of $\cA$ to Bob, having Bob continue feeding tokens into $\cA$ based on $y$, and finally querying $\cA$ for an answer to $P$, based on which Bob can give a good output in the communication game. When this works, it follows that the space used by $\cA$ must be at least the one-way (and perhaps randomized) communication complexity of the game. Note, however, that this style of argument where it is possible to solve the game by querying the algorithm only once, is also applicable to an oblivious adversary setting. Therefore, it cannot prove a lower bound any higher than the standard streaming complexity of $P$.

The way to obtain stronger lower bounds by using the purported adversarial robustness of $\cA$ is to design communication protocols where Bob, after receiving Alice's message, proceeds to query $\cA$ repeatedly, feeding tokens into $\cA$ based on answers to such queries. In fact, in the communication games we shall use for our reductions, Bob will not have any input at all and the goal of the game will be for Bob to recover information about Alice's input, perhaps indirectly. It should be clear that the lower bound for the \mif problem, outlined in \Cref{sec:related}, can be formalized in this manner. For our main lower bound (\Cref{thm:lb:preview}), we use a communication game that can be seen as a souped-up version of \mif.

\mypar{The Subset-Avoidance Problem}
Recall the \textsc{subset-avoidance} problem described in \Cref{sec:results} and denote it $\avoid(t,a,b)$. To restate: Alice is given a set $A\subseteq [t]$ of size $a$ and must induce Bob to output a set $B\subseteq [t]$ of size $b$ such that $A \cap B = \varnothing$. The one-way communication complexity of this game can be lower bounded from first principles. Since each output of Bob is compatible with only $\binom{t-b}{a}$ possible input sets of Alice, she cannot send the same message on more than that many inputs. Therefore, she must be able to send roughly $\binom{t}{a}/\binom{t-b}{a}$ distinct messages for a protocol to succeed with high probability. The number of bits she must communicate in the worst case is roughly the logarithm of this ratio, which we show is $\Omega(ab/t)$. Interestingly, this lower bound is tight and can in fact be matched by a deterministic protocol, as shown in \Cref{lem:avoid-det}.

In the sequel, we shall need to consider a direct sum version of this problem that we call $\kavoid(t,a,b)$, where Alice has a list of $k$ subsets and Bob must produce his own list of subsets, with his $i$th avoiding the $i$th subset of Alice. We extend our lower bound argument to show that the one-way complexity of $\kavoid(t,a,b)$ is $\Omega(kab/t)$.

\mypar{Using Graph Coloring to Solve Subset-Avoidance}
To explain how we reduce the \kavoid problem to graph coloring, we focus on a special case of \Cref{thm:lb:preview} first. Suppose we have an adversarially robust $(\Delta+1)$-coloring streaming algorithm $\cA$. We describe a protocol for solving $\avoid(t,a,b)$. Let us set $t=\binom{n}{2}$ to have the universe correspond to all possible edges of an $n$-vertex graph. Suppose Alice's set $A$ has size $a \approx n^2/8$. We show that, given a set of $n$ vertices, Alice can use public randomness to randomly map her elements to the set of vertex-pairs so that the corresponding edges induce a graph $G$ that, w.h.p., has max-degree $\Delta \approx n/4$.  Alice proceeds to feed the edges of $G$ into $\cA$ and then sends Bob the state of $\cA$.

Bob now queries $\cA$ to obtain a $(\Delta+1)$-coloring of $G$. Then, he pairs up like-colored vertices to obtain a maximal pairing. Observe that he can pair up all but at most one vertex from each color class. Thus, he obtains at least $(n-\Delta-1)/2$ such pairs. Since each pair is monochromatic, they don't share an edge, and hence, Bob has retrieved $(n-\Delta-1)/2$ missing edges that correspond to elements absent in Alice's set. Since Alice used public randomness for the mapping, Bob knows exactly which elements these are. He now forms a matching with these pairs and inserts the edges to $\cA$. Once again, he queries $\cA$ to find a coloring of the modified graph. Observe that the matching can increase the max-degree of the original graph by at most~$1$. Therefore, this new coloring uses at most $\Delta+2$ colors. Thus, Bob would retrieve at least $(n-\Delta-2)/2$ new missing edges. He again adds to the graph the matching formed by those edges and queries $\cA$. It is crucial to note here that he can repeatedly do this and expect $\cA$ to output a correct coloring because of its adversarial robustness. Bob stops once the max-degree reaches $n-1$, since now the algorithm can color each vertex with a distinct color, preventing him from finding a missing edge.

Summing up the sizes of all the matchings added by Bob, we see that he has found $\Theta((n-\Delta)^2)$ elements missing from Alice's set. Since $\Delta \approx n/4$, this is $\Theta(n^2)$. Thus, Alice and Bob have solved the $\avoid(t,a,b)$ problem where $t=\binom{n}{2}$ and $a,b=\Theta(n^2)$. As outlined above, this requires $\Omega(ab/t)=\Omega(n^2)$ communication. Hence, $\cA$ must use at least $\Omega(n^2)=\Omega(n\Delta)$ space.

With some further work, we can generalize the above argument to work for any value of $\Delta$ with $1 \le \Delta \le n/2$. For this generalization, we use the communication complexity of $\kavoid(t,a,b)$ for suitable parameter settings. With more rigorous analysis, we can further generalize the result to apply not only to $(\Delta+1)$-coloring algorithms but to any $f(\Delta)$-coloring algorithm. That is, we can prove \Cref{thm:lower-bound-core}. 

\subsection{Upper Bound Techniques} \label{sec:techniques-ub}

It is useful to outline our algorithms in an order different from the presentation in \Cref{sec:ubs}.

\mypar{A Sketch-Switching-Based $\bm{O(\Delta^2)}$-Coloring} The main challenge in designing an adversarially robust coloring algorithm is that the adversary can compel the algorithm to change its output at every point in the stream: he queries the algorithm, examines the returned coloring, and inserts an edge between two vertices of the same color.  Indeed, the sketch switching framework of \cite{BenEliezerJWY20} shows that for \emph{function estimation}, one can get around this power of the adversary as follows. Start with a basic (i.e., oblivious-adversary) sketch for the problem at hand. Then, to deal with an adaptive adversary, run multiple independent basic sketches in parallel, changing outputs only when forced to because the underlying function has changed significantly. More precisely, maintain $\lambda$ independent parallel sketches where
$\lambda$ is the \emph{flip number}, defined as the maximum number of times the function value can change by the desired approximation factor over the course of the stream. Keep track of which sketch is currently being used to report outputs to the adversary. Upon being queried, re-use the most recently given output unless forced to change, in which case discard the current sketch and switch to the next in the list of $\lambda$ sketches. Notice that this keeps the adversary oblivious to the randomness being used to compute future outputs: as soon as our output reveals any information about the current sketch, we discard it and never use it again to process a stream element.

This way of switching to a new sketch only when forced to ensures that $\lambda$ sketches suffice, which is great for function estimation. However, since a graph coloring output can be forced to change at every point in a stream of length $m$, naively implementing this idea would require $m$ parallel sketches, incurring a factor of $m$ in space. We have to be more sophisticated. We combine the above idea with a chunking technique so as to reduce the number of times we need to switch sketches. 

Suppose we split the $m$-length stream into $k$ chunks, each of size $m/k$. We initialize $k$ parallel sketches of a standard streaming $(\Delta+1)$-coloring algorithm $\cC$ to be used one at a time as each chunk ends. We store (buffer) an entire chunk explicitly and when we reach its end, we say we have reached a ``checkpoint,'' use a fresh copy of $\cC$ to compute a $(\Delta+1)$-coloring of the entire graph at that point, delete the chunk from our memory, and move on to store the next chunk. When a query arrives, we deterministically compute a $(\Delta+1)$-coloring of the partial chunk in our buffer and ``combine'' it with the coloring we computed at the last checkpoint. The combination uses at most $(\Delta+1)^2 = O(\Delta^2)$ colors. Since a single copy of $\cC$ takes $\tO(n)$ space, the total space used by the sketches is $\tO(nk)$. Buffering a chunk uses an additional $\tO(m/k)$ space. Setting $k$ to be $\sqrt{m/n}$, we get the total space usage to be $\tO(\sqrt{mn})=\tO(n\sqrt{\Delta})$, since $m=O(n\Delta)$. 

Handling edge deletions is more delicate. This is because we can no longer express the current graph as a union of $G_1$ (the graph up to the most recent checkpoint) and $G_2$ (the buffered subgraph) as above. A chunk may now contain an update that deletes an edge which was inserted before the checkpoint, and hence, is not in store. Observe, however, that deleting an edge doesn't violate the validity of a coloring. Hence, if we ignore these edge deletions, the only worry is that they might substantially reduce the maximum degree $\Delta$ causing us to use many more colors than desired. Now, note that if we have a $(\Delta_1+1)$-coloring at the checkpoint, then as long as the current maximum degree $\Delta$ remains above $\Delta_1/2$, we have a $2\Delta$-coloring in store. Hence, combining that with a $(\Delta+1)$-coloring of the current chunk gives an $O(\Delta^2)$-coloring. Furthermore, we can keep track of the maximum degree of the graph using only $\tO(n)$ space and detect the points where it falls below half of what it was at the last checkpoint. We declare each such point as a new ``ad hoc checkpoint,'' i.e., use a fresh sketch to compute a $(\Delta+1)$-coloring there. Since the max-degree can decrease by a factor of $2$ at most $\log n$ times, we show that it suffices to have only $\log n$ times more parallel sketches initialized at the beginning of the stream. This incurs only an $O(\log n)$-factor overhead in space. We discuss the algorithm and its analysis in detail in \Cref{alg:deltasq-turnstile} and \Cref{lem:turnstile-ub-square} respectively.

To generalize the above to an $O(\Delta^k)$-coloring in $\tO(n\Delta^{1/k})$ space, we use recursion in a manner reminiscent of streaming coreset construction algorithms. Split the stream into $\Delta^{1/k}$ chunks, each of size $n\Delta^{1-1/k}$. Now, instead of storing a chunk entirely and coloring it deterministically, we can recursively color it with $\Delta^{k-1}$ colors in $O(n\Delta^{1/k})$ space and combine the coloring with the $(\Delta+1)$-coloring at the last checkpoint. The recursion makes the analysis of this algorithm even more delicate, and careful work is needed to argue the space usage and to properly handle deletions in the turnstile setting. The details appear in \Cref{thm:turnstile-ub-general}.

\mypar{A Palette-Sparsification-Based $\bm{O(\Delta^3)}$-Coloring} This algorithm uses a different approach to the problem of the adversary forcing color changes. It ensures that, every time an an edge is added, one of its endpoints is
randomly recolored, where the color is drawn uniformly from a set $C \setm K$ of colors, where $C$ is determined by the degree of the endpoint, and $K$ is the set of colors currently held by neighboring vertices. Let $R_v$ denote the random string that drives this color-choosing process at vertex $v$. When the adversary inserts an edge $\{u,v\}$, the algorithm uses $R_u$ and $R_v$ to determine whether this edge could with significant probability end up with the same vertex color on both ends in the future. If so,
the algorithm stores the edge; if not, it can be ignored entirely. It will turn out that when the number of colors is set to establish an $O(\Delta^3)$-coloring, only an $\tO(1 / \Delta)$ fraction of edges need to be stored, so the algorithm only needs to store $\tO(n)$ bits of data related to the input. The proof of this storage bound has to contend with an adaptive adversary. We do so by first arguing that despite this adaptivity, the adversary cannot cause the algorithm to use more storage than the worst oblivious adversary could have. We can then complete the proof along traditional lines, using concentration bounds. The details appear in \Cref{alg:ub-delta3} and \Cref{thm:ub-delta3}.

There is a technical caveat here. The random string $R_v$ used at each vertex $v$ is about $\tO(\Delta)$ bits long. Thus, the algorithm can only be called semi-streaming if we agree that these $\tO(n\Delta)$ random bits do not count towards the storage cost. In the standard streaming setting, this ``randomness cost'' is not a concern, for we can use the standard technique of invoking Nisan's space-bounded pseudorandom generator~\cite{Nisan90} to argue that the necessary bits can be generated on the fly and never stored. Unfortunately, it is not clear that this transformation preserves adversarial robustness. Despite this caveat, the algorithmic result is interesting as a contrast to our lower bounds, because the lower bounds do apply even in a model where random bits are free, and only actually computed input-dependent bits count towards the space complexity.

%% file: prelims.tex

\section{Preliminaries} \label{sec:prelims}

\mypar{Defining Adversarial Robustness}
For the purposes of this paper, a ``streaming algorithm'' is always
one-pass and we always think of it as working against an adversary. In the
\emph{standard streaming} setting, this adversary is oblivious to the
algorithm's actual run. This can be thought of as a special case of the setup
we now introduce in order to define \emph{adversarially robust streaming}
algorithms.

Let $\cU$ be a universe whose elements are called tokens. A data stream is a
sequence in $\cU^*$. A data streaming problem is specified by a relation $f
\subseteq \cU^* \times \cZ$ where $\cZ$ is some output domain: for each input
stream $\sigma \in \cU^*$, a valid solution is any $z \in \cZ$ such that
$(\sigma,z) \in f$. A randomized streaming algorithm $\cA$ for $f$ running in
$s$ bits of space and using $r$ random bits is formalized as a triple
consisting of (i)~a function $\init \colon \b^r \to \b^s$, (ii)~a function
$\process \colon \b^s \times \cU \times \b^r \to \b^s$, and (iii)~a function
$\query \colon \b^s \times \b^r \to \cZ$. Given an input stream $\sigma =
(x_1,\ldots,x_m)$ and a random string $R \in_R \b^r$, the algorithm starts in
state $w_0 = \init(R)$, goes through a sequence of states $w_1,\ldots,w_m$,
where $w_i = \process(w_{i-1}, x_i, R)$, and provides an output $z =
\query(w_m, R)$. The algorithm is $\delta$-error in the standard sense if
$\Pr_R[(\sigma,z) \in f] \ge 1-\delta$.

To define adversarially robust streaming, we set up a game between two
players: Solver, who runs an algorithm as above, and Adversary, who adaptively
generates a stream $\sigma = (x_1,\ldots,x_m)$ using a next-token function
$\next \colon \cZ^* \to \cU$ as follows. With $w_0, \ldots, w_m$ as above, put
$z_i = \query(w_i, R)$ and $x_i = \next(z_0,\ldots,z_{i-1})$. In words,
Adversary is able to query the algorithm at each point of time and can compute
an arbitrary \emph{deterministic} function of the history of outputs provided
by the algorithm to generate his next token. Fix (an upper bound on) the
stream length $m$. Algorithm $\cA$ is \emph{$\delta$-error
adversarially robust} if 
\[
  \forall \text{ function \next}:~ 
    \Pr_R[ \forall\, i \in [m]:~ 
      ((x_1,\ldots,x_i), z_i) \in f] \ge 1-\delta \,.
\]
In this work, we prove lower bounds for algorithms
that are only required to be $O(1)$-error adversarially robust. On the other
hand, the algorithms we design will achieve vanishingly small error of the
form $1/\poly(m)$ and moreover, they will be able to detect when they are
about to err and can abort at that point.

\mypar{Graph Streams and the Coloring Problem}
Throughout this paper, an \emph{insert-only graph stream} describes an undirected graph on the vertex set $[n]$, for some fixed $n$ that is known in advance, by listing its edges in some order: each token is an edge. A \emph{strict graph turnstile stream} describes an evolving graph $G$ by using two types of tokens---$\insedge(\{u,v\})$, which causes $\{u,v\}$ to be added to $G$, and $\deledge(\{u,v\})$, which causes $\{u,v\}$ to be removed---and satisfies the promises that each insertion is of an edge that was not already in $G$ and that each deletion is of an edge that was in $G$. When we use the term ``graph stream'' without qualification, it should be understood to mean an insert-only graph stream, unless the context suggests that either flavor is acceptable.

In this context, a semi-streaming algorithm is one that runs in $\tO(n) := O(n \polylog n)$ bits of space.

In the $K$-coloring problem, the input is a graph stream and a valid answer to a query is a vector in $[K]^n$ specifying a color for each vertex such that no two adjacent vertices receive the same color. The quantity $K$ may be given as a function of some graph parameter, such as the maximum degree $\Delta$. In reading the results in this paper, it will be helpful to think of $\Delta$ as a growing but sublinear function of $n$, such as $n^\alpha$ for $0 < \alpha < 1$. Since an output of the $K$-coloring problem is a $\Theta(n \log K)$-sized object, we think of a semi-streaming coloring algorithm running in $\tO(n)$ space as having ``essentially optimal'' space usage.

\mypar{One-Way Communication Complexity}
In this work, we shall only consider a special kind of two-player communication game: one where all input belongs to the speaking player Alice and her goal is to induce Bob to produce a suitable output. Such a game, $g$, is given by a relation $g \in \cX \times \cZ$, where $\cX$ is the input domain and $\cZ$ is the output domain. In a protocol $\Pi$ for $g$, Alice and Bob share a random string $R$. Alice is given $x \in \cX$ and sends Bob a message $\msg(x,R)$. Bob uses this to compute an output $z = \out(\msg(x,R))$. We say that $\Pi$ solves $g$ to error $\delta$ if $\forall\, x\in \cX:~ \Pr_R[(x,z) \in g] \ge 1-\delta$. The communication cost of $\Pi$ is $\cost(\Pi) := \max_{x,R} \text{length}(\msg(x,R))$. The (one-way, randomized, public-coin) $\delta$-error communication complexity of $g$ is $\R^\to_\delta(g) := \min\{\cost(\Pi):\, \Pi$ solves $g$ to error $\delta\}$. 

If $\Pi$ never uses $R$, it is deterministic. Minimizing over zero-error deterministic protocols gives us the one-way deterministic communication complexity of $g$, denoted $\D^\to(g)$.

\mypar{A Result on Random Graphs}
During the proof of our main lower bound (in \Cref{sec:avoid-to-coloring}), we shall need the following basic lemma on the maximum degree of a random graph.

\begin{lemma}\label{lem:random-graph-max-degree} Let $G$ be a graph with $M$ edges and $n$ vertices, drawn uniformly at random. Define $\Delta_G$ to be its maximum degree. Then for $0 \le \eps \le 1$:
\begin{align}
    \Pr\left[ \Delta_G \ge \frac{2 M}{n} (1 + \eps) \right] \le 2 n\exp\left(- \frac{\eps^2}{3} \cdot \frac{2M}{n} \right)\,.
\end{align}
\end{lemma}

\begin{proof}
  Let $G(n,m)$ be the uniform distribution over graphs with $m$ edges and $n$ vertices. Observe the monotonicity
  property that for all $m\in\NN$, $\Pr_{G \sim G(n,m)}[\Delta_G \ge C] \le \Pr_{G \sim G(n,m+1)}[\Delta_G \ge C]$. Next,
  let $H(n,p)$ be the distribution over graphs on $n$ vertices in which each edge is included with probability $p$,
  independently of any others, and let $e(G)$ be the number of edges of a given graph $G$. Then with $p = M / \binom{n}{2}$,
  \begin{align*}
    \Pr_{G \sim G(n,M)}[\Delta_G \ge C]
    &= \Pr_{G \sim H(n,p)}[\Delta_G \ge C \mid e(G) = M]
    \le \Pr_{G \sim H(n,p)}[\Delta_G \ge C \mid e(G) \ge M] && \lhd~\text{by monotonicity} \\
    &\le  \frac{\Pr_{G \sim H(n,p)}[\Delta_G \ge C]}{\Pr_{G \sim H(n,p)}[e(G) \ge M]}
    \le 2 \Pr_{G \sim H(n,p)}[\Delta_G \ge C] \,.
  \end{align*}
  The last step follows from the well-known fact that the median of a binomial distribution
  equals its expectation when the latter is integral; hence $\Pr_{G \sim H(n,p)}[e(G) \ge M] \ge 1/2$.

  Taking $C = (2M/n) (1+\eps)$ and using a union bound and Chernoff's inequality,
  \begin{align*}
    \Pr_{G \sim H(n,p)}\left[ \Delta_G \ge \frac{2M}{n} (1+\eps) \right]
    &\le \sum_{x \in V(G)} \Pr_{G \sim H(n,p)}\left[\deg_G(x) \ge \frac{2M}{n} (1+\eps)\right]
    \le n \exp\left(- \frac{\eps^2}{3} \cdot \frac{2M}{n} \right) \,.
    \qedhere
  \end{align*}
\end{proof}

\mypar{Algorithmic Results From Prior Work}
Our adversarially robust graph coloring algorithms in \Cref{sec:ub-sketchswitch} will use, as subroutines, some previously known standard streaming algorithms for coloring. We summarize the key properties of these existing algorithms.

\begin{fact}[Restatement of \cite{AssadiCK19}, Result~2]\label{fact:ackalgo}
There is a randomized turnstile streaming algorithm for $(\Delta+1)$-coloring a graph with max-degree $\Delta$ in the oblivious adversary setting that uses $\tO(n)$ bits of space and $\tO(n)$ random bits. The failure probability can be made at most $1/n^p$ for any large constant $p$. 
\qed
\end{fact}

In the adversarial model described above, we need to answer a query after each stream update. The algorithm mentioned in \Cref{fact:ackalgo} or other known algorithms using ``about'' $\Delta$ colors (e.g., \cite{BeraCG20}) use at least $\tilde{\Theta}(n)$ post-processing time in the worst case to answer a query. Hence, using such algorithms in the adaptive adversary setting might be inefficient. We observe, however, that at least for insert-only streams, there exists an algorithm that is efficient in terms of both space and time. This is obtained by combining the algorithms of \cite{BeraCG20} and \cite{HenzingerP20} (see the discussion towards the end of \Cref{sec:ub-sketchswitch} for details).

\begin{fact}\label{fact:bcgsmallupdate}
In the oblivious adversary setting, there is a randomized streaming algorithm that receives a stream of edge insertions of a graph with max-degree $\Delta$ and degeneracy $\kappa$ and maintains a proper coloring of the graph using $\kappa(1+\eps)\leq \Delta(1+\eps)$ colors, $\tO(\eps^{-2}n)$ space, and $O(1)$ amortized update time. The failure probability can be made at most $1/n^p$ for any large constant $p$. \qed
\end{fact}

%% file: lb.tex

\section{Hardness of Adversarially Robust Graph Coloring}

In this section, we prove our first major result, showing that graph coloring
is significantly harder when working against an adaptive adversary than it is
in the standard setting of an oblivious adversary. We carry out the proof plan
outlined in \Cref{sec:techniques-lb}, first describing and analyzing our novel
communication game of \textsc{subset-avoidance} (henceforth, \avoid) and then
reducing the \avoid problem to robust coloring.

\subsection{The Subset Avoidance Problem} \label{sec:avoid}

Let $\avoid(t,a,b)$ denote the following one-way communication game.
\begin{itemize}[itemsep=1pt]
    \item Alice is given $S \subseteq [t]$ with $|S|=a$;
    \item Bob must produce $T \subseteq [t]$ with $|T|=b$ for which $T$ is disjoint from $S$.
\end{itemize}
Let $\kavoid(t,a,b)$ be the problem of simultaneously solving $k$ instances of $\avoid(t,a,b)$.

\begin{lemma}\label{lem:disjrec-lb}
  The public-coin $\delta$-error communication complexity of $\kavoid(t,a,b)$ is bounded thus:
  \begin{align}
    \R^\to_\delta(\kavoid(t,a,b))
    &\ge \log{(1-\delta)} + k \log{\left( \binom{t}{a} \Big/ \binom{t-b}{a} \right)} \label{eq:disjrec-lb} \\
    &\ge \log{(1-\delta)} + k a b / {(t \ln 2)} \label{eq:avoid-lb-clean} \,.
  \end{align}
\end{lemma}

\begin{proof}
  Let $\Pi$ be a $\delta$-error protocol for $\kavoid(t,a,b)$ and let $d = \cost(\Pi)$, as defined in \Cref{sec:prelims}. Since, for each input $(S_1,\ldots,S_k) \in \binom{[t]}{a}^k$, the error probability of $\Pi$ on that input is at most $\delta$, there must exist a fixing of the random coins of $\Pi$ so that the resulting deterministic protocol $\Pi'$ is correct on all inputs in a set
  \[
    \cC \subseteq \binom{[t]}{a}^k \,, \quad\text{with } |\cC| \ge (1-\delta) \binom{t}{a}^k \,.
  \]
  The protocol $\Pi'$ is equivalent to a function $\phi \colon \cC \to \binom{[t]}{b}^k$ where 
  \begin{itemize}[itemsep=1pt]
    \item the range size $|\image(\phi)| \le 2^d$, because $\cost(\Pi) \le d$, and
    \item for each $(S_1,\ldots,S_k) \in \cC$, the tuple $(T_1, \ldots, T_k) := \phi((S_1,\ldots,S_k))$ is a correct output for Bob, i.e., $S_i \cap T_i = \varnothing$ for each $i$.
  \end{itemize}

  For any fixed $(T_1,\ldots,T_k) \in \binom{[t]}{b}^k$, the set of all $(S_1,\ldots,S_k) \in \binom{[t]}{a}^k$ for which each coordinate $S_i$ is disjoint from the corresponding $T_i$ is precisely the set $\binom{[t] \setm T_1}{S_1} \times \cdots \times \binom{[t] \setm T_k}{S_k}$. The cardinality of this set is exactly $\binom{t - b}{a}^k$.  Thus, for any subset $\cD$ of $\binom{[t]}{b}^k$, it holds that $\left|\cC \cap \phi^{-1}(\cD)\right| \le \binom{t - b}{a}^k |\cD|$. Consequently,
  \begin{align*}
      (1-\delta) \binom{t}{a}^k \le |\cC| = |\phi^{-1}(\image(\phi))|
      \le \binom{t - b}{a}^k |\image(\phi)| \le \binom{t - b}{a}^k 2^d \,,
  \end{align*}
  which, on rearrangement, gives \cref{eq:disjrec-lb}.

  To obtain \cref{eq:avoid-lb-clean}, we note that
  \begin{align}
      \binom{t}{a} \Big/ \binom{t - b}{a}
          &= \frac{t! a! (t - a - b)!}{(t-a)! a! (t-b)!}
          = \frac{t \cdot (t - 1) \cdots (t - a + 1)}{(t - b) \cdot (t - b - 1) \cdots (t - a - b + 1)} \nonumber\\
          &\ge \left(\frac{t}{t-b}\right)^a 
          = \left(\frac{1}{1-b/t}\right)^a 
          > e^{ab/t} \,, \label{eq:binratio-2}
  \end{align}
  which implies
  \[
      \log{(1-\delta)} + k \log{\left( \binom{t}{a} \Big/ \binom{t-b}{a} \right)} \ge \log{(1-\delta)} + k ab/ {(t \ln 2)} \,. \qedhere
  \]
\end{proof}

Since our data streaming lower bounds are based on the \kavoid problem, it is important to verify that we are not analyzing its communication complexity too loosely. To this end, we prove the following result, which says that the lower bound in \Cref{lem:disjrec-lb} is close to being tight. In fact, a nearly matching upper bound can be obtained deterministically.

\begin{lemma} \label{lem:avoid-det}
For any $t \in \NN$, $0 < a + b \le t$, the deterministic complexity of $\avoid(t,a,b)$ is bounded thus:
\begin{align}
    \D^\to(\avoid(t,a,b)) \le \log \left(\binom{t}{a} \Big/ \binom{t-b}{a}  \right) + \log\left(\ln\binom{t}{a}\right) + 2 \,. \label{eq:avoid-det-ubound}
\end{align}
\end{lemma}

\begin{proof}
   We claim there exists an ordered collection $\cR$ of $z := \big\lceil \big(\binom{t}{a} \big/ \binom{t-b}{a}\big) \ln\binom{t}{a} \big\rceil$ subsets of $[t]$ of size $b$, with the property that for each $S \in \binom{[t]}{a}$, there exists a set $T$ in $\cR$ which is disjoint from $S$. In this case, Alice's protocol is, given a set $S \in \binom{[t]}{a}$, to send the index $j$ of the first set $T$ in $\cR$ which is disjoint from $S$; Bob in turn returns the $j$th element of $\cR$. The number of bits needed to communicate such an index is at most $\ceil{\log z}$, implying \cref{eq:avoid-det-ubound}.
   
   We prove the existence of such an $\cR$ by the probabilistic method. Pick a subset $\cQ\subseteq \binom{[t]}{b}$ of size $z$ uniformly at random. For any $S \in \binom{[t]}{a}$, define $\cO_S$ to be the set of subsets in $\binom{[t]}{b}$ which are disjoint from $S$; observe that $|\cO_S| = \binom{t-a}{b}$. Then $\cQ$ has the desired property if for all $S \in \binom{[t]}{a}$, it overlaps with $\cO_S$. As
   \begin{align*}
     \Pr\left[\exists S \in \binom{[t]}{a} : \cQ \cap \cO_S = \varnothing\right] 
       &\le \sum_{S \in  \binom{[t]}{a}} \Pr\left[\cQ \cap \cO_S =\varnothing \right] && \lhd~\text{by union bound}\\
       & = \sum_{S \in  \binom{[t]}{a}} \Pr\left[\cQ \in \binom{\binom{[t]}{b} \setm \cO_S }{z} \right] \\
       & = \sum_{S \in  \binom{[t]}{a}}\left( \binom{\binom{t}{b} - \binom{t - a}{b}}{z} \Big/ \binom{\binom{t}{b}}{z} \right) \\
       & < \binom{t}{a} \exp\left( - z \binom{t - a}{b} \Big/ \binom{t}{b} \right) && \lhd~\text{by \cref{eq:binratio-2}} \\
       & = \binom{t}{a} \exp\left( - z \binom{t - b}{a} \Big/ \binom{t}{a} \right) \,,
   \end{align*}
   setting $z = \big\lceil \big(\binom{t}{a} \big/ \binom{t-b}{a}\big) \ln\binom{t}{a} \big\rceil$ ensures the random set $\cQ$ fails to have the desired property with probability strictly less than 1. Let $\cR$ be a realization of $\cQ$ that does have the property.
\end{proof}


\subsection{Reducing Multiple Subset Avoidance to Graph Coloring} \label{sec:avoid-to-coloring}

Having introduced and analyzed the \avoid communication game, we are now ready to prove our main lower bound result, on the hardness of adversarially robust graph coloring.

\begin{theorem}[Main lower bound]
\label{thm:lower-bound-core}
  Let $L, n, K$ be integers with $2K \le n$, and $L + 1 \le K$,
  and $L \ge 12 \ln (4n)$.

  Assume there is an adversarially robust coloring algorithm $\cA$ for insert-only streams of $n$-vertex graphs which works as long as the input graph has maximum degree $\le L$, and maintains a coloring with $\le K$ colors so that all colorings are correct with probability $\ge 1/4$.
  Then $\cA$ requires at least $C$ bits of space, where
  \begin{align*}
    C \ge \frac{1}{40 \ln 2} \cdot \frac{n L^2}{K} - 3 \,.
  \end{align*}
\end{theorem}

\begin{proof}
Given an algorithm $\cA$ as specified, we can construct a public-coin protocol
to solve the communication problem $\avoid^{\floor{n/(2K)}}(\binom{2 K}{2}, \floor{LK/4}, \floor{L/2}\ceil{K/2} )$ using exactly as much communication as $\cA$ requires storage space. The protocol for the more basic problem $\avoid(\binom{2 K}{2}, \floor{LK/4}, \floor{L/2}\ceil{K/2} )$ is described in \Cref{alg:recovery}.

\begin{algorithm}[!ht]
  \caption{Protocol for $\avoid(\binom{2 K}{2},\floor{LK/4},\floor{L/2}\ceil{K/2})$
    \label{alg:recovery}}
  \begin{algorithmic}[1]
    \Statex \textbf{Require:} Algorithm $\cA$ that colors graphs up to maximum degree $L$, always using $\le K$ colors
    \Let{$R$}{publicly random bits to be used by $\cA$}
    \Let{$\pi$}{publicly random permutation of $\{1, \ldots, \binom{2 K}{2}\}$, drawn uniformly}
    \Let{$e_1,\ldots,e_{\binom{2K}{2}}$}{an enumeration of the edges of the complete graph on $2K$ vertices}
    \Statex
    \Function{Alice}{S}:
    \Let{$Z$}{$\cA$::INIT($R$), the initial state of $\cA$}
      \For{$i$ from $1$ to $\binom{2 K}{2}$}
        \If{$\pi_i \in S$}
          \Let{$Z$}{$\cA$::INSERT(Z, $R$, $e_i$) }
        \EndIf
      \EndFor
      \State \Return{$Z$}
    \EndFunction
    \Statex
    \Function{Bob}{$Z$}:
      \Let{$J$}{empty list}
      \For{$i$ from $1$ to $\floor{L/2}$}
        \Let{$\clr$}{$\cA$::QUERY($Z$, $R$)}\label{step:iter-get-coloring}
        \Let{$M$}{maximal pairing of like-colored vertices, according to $\clr$}
        \For{each pair $\{u,v\} \in M$}
            \Let{$Z$}{$\cA$::INSERT($Z$,
            $R$, $\{u,v\}$) } \Comment{$M$ is turned into a matching and inserted}\label{step:update-z-with-m}
        \EndFor 
        \State $J \leftarrow J \cup M$
      \EndFor
      \If{$\text{length}(J) \le \floor{L/2}\ceil{K/2}$}
        \State \Return{fail}
      \Else
        \Let{$T$}{ $\{\pi_i : e_i \in \text{first $\floor{L/2}\ceil{K/2}$ edges of $J$} \}$ } \label{step:graph-to-set}
        \State \Return{$T$}
      \EndIf
    \EndFunction
  \end{algorithmic}
\end{algorithm}

To use $\cA$ to solve $s := \floor{n/2K}$ instances of $\avoid$, we pick $s$ disjoint subsets $V_1,\ldots,V_s$ of the vertex set $[n]$, each of size $2K$. A streaming coloring algorithm on the vertex set $[2K]$ with degree limit $L$ and using at most $K$ colors can be implemented by relabeling the vertices in $[2K]$ to the vertices in some set $V_i$ and using $\cA$. This can be done $s$ times in parallel, as the sets $(V_i)_{i=1}^{s}$ are disjoint. Note that a coloring of the entire graph on vertex set $[n]$ using $\le K$ colors is also a $K$-coloring of the $s$ subgraphs supported on $V_1,\ldots,V_s$. To minimize the number of color queries made, \Cref{alg:recovery} can be implemented by alternating between adding elements from the matching $M$ in each instance (for \Cref{step:update-z-with-m}), and making single color queries to the $n$-vertex graph (for \Cref{step:iter-get-coloring}).

The guarantee that $\cA$ uses fewer than $K$ colors depends on the input graph
stream having maximum degree at most $L$. In Bob's part of the protocol,
adding a matching to the graph only increases the maximum degree of the
graph represented by $Z$ by at most one; since he does this $\floor{L/2}$ times, in order for the maximum degree of the graph represented by $Z$ to remain at most $L$, we would like the random graph Alice inserts into the algorithm to have maximum degree $\le L/2 \le L - \floor{L/2}$. By \Cref{lem:random-graph-max-degree}, the probability that, given some $i$, this random graph on $V_i$ has maximum degree $\Delta_i \ge L/2$ is
\begin{align*}
    \Pr\left[ \Delta_i \ge \frac{L}{4} (1 + 1) \right] \le 4 K e^{ -L / 12 } \,.
\end{align*}
Taking a union bound over all $s$ graphs, we find that
\begin{align*}
    \Pr\left[ \max_{i \in [s]} \Delta_i \ge L /2 \right] \le 4 K \left\lfloor \frac{n}{2K} \right\rfloor  e^{ -L / 12 } \le 2 n  e^{ -L / 12 } \,.
\end{align*}
We can ensure that this happens with probability at most $1/2$ by requiring $L \ge 12 \ln(4n)$.

If all the random graphs produced by Alice have maximum degree $\le L/2$, and the $\floor{L/2}$ colorings requested by the protocol are all correct, then we will show that Bob's part of the protocol recovers at least $\floor{L/2}\ceil{K/2}$ edges for each instance. Since the algorithm $\cA$'s random bits $R$ and permutation random bits $\pi$ are independent, the probability that the the maximum degree is low and the algorithm gives correct colorings on graphs of maximum degree at most $L$ is $\ge (1/2) \cdot (1/4) = 1/8$.

The list of edges that Bob inserts (\Cref{step:update-z-with-m}) are fixed functions of the query output of $\cA$ on its state $Z$ and random bits $R$. None of the edges can already have been inserted by Alice or Bob, since each edge connects two vertices which have the same color.
Because these edges only depend on the query output of $\cA$, conditioned on this query output they are independent of $Z$ and $R$. This ensures that $\cA$'s correctness guarantee against an adversary applies here, and thus the colorings reported on \Cref{step:iter-get-coloring} are correct.

Assuming all queries succeed, and the initial graph that Alice added has maximum degree $\le L/2$, for each $i \in [\floor{L/2}]$, the coloring produced will have at most $K$ colors. Let $B$ be the set of vertices covered by the matching $M$, so that $[2K]\setm B$ are the unmatched vertices. Since no pair of unmatched vertices can have the same color, $\left|[2K]\setm B\right| \le K$. This implies $|B| \ge K$, and since $|M| = |B| / 2$ is an integer, we have $|M| \ge \ceil{K/2}$. Thus each \textbf{for} loop iteration will add at least $\ceil{K/2}$ new edges to $J$. The final value of the list $J$ will contain at least $\floor{L/2}\ceil{K/2}$ edges that were not added by Alice; \Cref{step:graph-to-set} converts the first $\floor{L/2}\ceil{K/2}$ of these to elements of $\{1,\ldots,\binom{2K}{2}\}$ not in the set $S$ given to Alice.

Finally, by applying \Cref{lem:disjrec-lb}, we find that the communication $C$ needed
to solve $s$ independent copies of $\avoid(\binom{2K}{2},\floor{LK/4},\floor{L/2}\ceil{K/2})$ with failure probability $\le 7/8$ satisfies
\begin{align*}
    C &\ge \log\left(1 - \frac{7}{8}\right) + \left\lfloor\frac{n}{2 K}\right\rfloor \frac{ \floor{LK/4} \cdot \floor{L/2}\ceil{K/2}} { \binom{2K}{2} \ln 2} \\
    &\ge \frac{n}{4K} \frac{L^2 K^2 / 20}{\frac{1}{2} (2 K)^2\ln 2} - 3
     \ge \frac{n L^2}{40 K \ln 2} - 3 \,,
\end{align*}
where we used $K > L \ge 12 \ln(4n) \ge 12\ln 4$ to conclude
$\floor{LK/4}\floor{L/2}\ceil{K/2} \ge (LK)^2 / 20$.
\end{proof}

Applying the above \Cref{thm:lower-bound-core} with ``$K = f(L)$,'' we immediately obtain the following corollary, which highlights certain parameter settings that are particularly instructive.

\begin{corollary}\label{thm:lb-monotonic}
  Let $f$ be a monotonically increasing function, and $L$ an integer for which $L = \Omega(\log n)$ and $f(L) \le n/2$. Let $\cA$ be a coloring algorithm which works for graphs of maximum degree up to $L$; which at any point in time uses $\le f(\Delta)$ colors, where $\Delta$ is the current graph's maximum degree; and which has total failure probability $\le 3/4$ against an adaptive adversary.
  Then the number of bits $S$ of space used by $\cA$ is lower-bounded as $S = \Omega(n L^2 / f(L))$. In particular:
  \begin{itemize}[itemsep=1pt]
      \item If $f(\Delta) = \Delta + 1$---or, more generally, $f(\Delta) = O(\Delta)$---then $S = \Omega(n L)$ space is needed.
      \item To ensure $S = \tO(n)$ space, $f(\Delta) = \tOmega(\Delta^2)$ is needed.
      \item If $f(L) = \Theta(n)$, then $S = \Omega( L^2 )$. \qed
  \end{itemize}
\end{corollary}

%% file: ubs.tex

\section{Upper Bounds: Adversarially Robust Coloring Algorithms} \label{sec:ubs}

We now turn to positive results. We show how to maintain a $\poly(\Delta)$-coloring of a graph in an adversarially robust fashion. We design two broad classes of algorithms. The first, described in \Cref{sec:ub-deltacube}, is based on palette sparsification as in~\cite{AssadiCK19, AlonA20}, with suitable enhancements to ensure robustness. The resulting algorithm maintains an $O(\Delta^3)$-coloring and uses $\tO(n)$ bits of working memory. As noted in \Cref{sec:techniques-ub}, the algorithm comes with the caveat that it requires a large pool of random bits: up to $\tO(n\Delta)$ of them. As also noted there, it makes sense to treat this randomness cost as separate from the space cost.

The second class of algorithms, described in \Cref{sec:ub-sketchswitch}, is built on top of the sketch switching technique of \cite{BenEliezerJWY20}, suitably modified to handle non-real-valued outputs. This time, the amount of randomness used is small enough that we can afford to store all random bits in working memory. These algorithms can be enhanced to handle strict graph turnstile streams as described in \Cref{sec:prelims}. For any such turnstile stream of length at most~$m$, we maintain an $O(\Delta^2)$-coloring using $\tO(\sqrt{nm})$ space. More generally, we maintain an $O(\Delta^k)$-coloring in $O(n^{1-1/k}m^{1/k})$ space for any $k \in \NN$. In particular, for insert-only streams, this implies an $O(\Delta^k)$-coloring in $O(n\Delta^{1/k})$ space.

\input{ub-deltacubed}

\input{ub-deltasquare}

%% file: ub-deltacubed.tex

\subsection{An Algorithm Based on Palette Sparsification}
\label{sec:ub-deltacube}

We proceed to describe our palette-sparsification-based algorithm. It
maintains a $3\Delta^3$-coloring of the input graph $G$, where $\Delta$ is the
evolving maximum degree of the input graph $G$. With high probability, it will
store only $O(n (\log n)^4) = \tO(n)$ bits of information about $G$; an easy
modification ensures that this bound is always maintained by having the
algorithm abort if it is about to overshoot the bound. 

The algorithm does need a large number of random bits---up to $O(n L (\log
n)^2)$ of them---where $L$ is the maximum degree of the graph at the end of
the stream or an upper bound on the same.  Due to the way the algorithm looks
ahead at future random bits, $L$ must be known in advance. 

The algorithm uses these available random bits to pick, for each vertex, $L$
lists of random color palettes, one at each of $L$ ``levels.'' The level-$i$
list at vertex $x$ is called $P^{i}_x$ and consists of $4\log n$ colors
picked uniformly at random with replacement from the set $[2i^2]$. The
algorithm tracks each vertex's degree. Whenever a vertex $x$ is recolored, its
new color is always of the form $(d,p)$, where $d = \deg(x)$ and $p \in
P^{d}_x$.  Thus, when the maximum degree in $G$ is $\Delta$, the only colors
that have been used are the initial default $(0,0)$ and colors from
$\bigcup_{i=1}^\Delta \{i\} \times [2 i^2]$. The total number of colors is
therefore at most $1 + \sum_{i = 1}^{\Delta} 2 i^2 \le 3 \Delta^3$.

The precise algorithm is given in \Cref{alg:ub-delta3}. 

\begin{algorithm}[!ht]
  \caption{Adversarially robust $3 \Delta^3$-coloring algorithm, assuming $0 < \Delta \le L$
    \label{alg:ub-delta3}}
 
  \begin{algorithmic}[1]
    \Statex \textbf{Input}: Stream of edges of a graph $G=(V,E)$, with maximum degree always $\le L$.
    \Statex
    \Statex \ul{\textbf{Random bits}:}
    \For{each vertex $x \in [n]$}
        \For{each $i \in [L]$}
            \Let{$P^{i}_x$}{list of $4\log n$ colors sampled u.a.r.~with replacement from $[2i^2]$}
        \EndFor
    \EndFor
    
    \Statex
    \Statex \ul{\textbf{Initialize}:}
    \For{each vertex $x \in [n]$}
        \Let{$\degr(x)$}{$0$} \Comment{tracks degree of $x$}
        \Let{$\clr(x)$}{$(0,0)$} \Comment{maintains color of $x$; in general $\in \bigcup_{i=1}^{L} \{i\} \times [2 i^2]$}
    \EndFor
    \Let{$A$}{empty list of edges}
    \Statex
    
    \Statex \ul{\textbf{Process}(edge $\{u,v\}$):}
    \Let{$\degr(u), \degr(v)$}{$\degr(u)+1, \degr(v)+1$}\label{step:inc-deg-u} \Comment{maintain vertex degrees}
    \Let{$k$}{$\max\{\degr(u), \degr(v)\}$}\label{step:min-conf-level}
    \For{$i$ from $k$ to $L$}\label{step:check-range} \Comment{store edges that might be needed in the future}
        \If{$P^{i}_u$ and $P^{i}_v$ overlap}\label{step:check-edge}
            \Let{$A$}{$A \cup \{\{u,v\}\}$}\label{step:add-edge}
        \EndIf
    \EndFor
    \Let{$\used$}{$\{\clr(w) : \{u,w\}\in A\}$}\label{step:collect-neighbors} \Comment{prepare to recolor vertex $u$: collect colors of neighbors}
    \For{$j$ from $1$ to $4 \log n$}
        \Let{$c$}{$(\degr(u), P^{\degr(u)}_u[j])$} \Comment{try the next color in the random list}
        \If{$c \notin \used$}\label{step:check-color-validity}
            \Let{$\clr(u)$}{$c$};~ \textbf{return} \label{step:stop-on-first-color}
        \EndIf
    \EndFor
    \State \textbf{abort}  \Comment{failed to find a color}
    \Statex
    
    \Statex \ul{\textbf{Query}(~):}
    \State \Return{the vector $\clr$}
  \end{algorithmic}
\end{algorithm}

\begin{lemma}[Bounding the failure probability] \label{lemma:recolor-fail-probability}
When an edge is added, recoloring one of its vertices succeeds with probability $\ge 1 - 1 / n^4$, regardless of the past history of the algorithm.
\end{lemma}

\begin{proof}
  The color for the endpoint $u$ is chosen and assigned in Lines
  \ref{step:collect-neighbors} through \ref{step:stop-on-first-color}. Let $d$
  be the value of $\degr(u)$ at that point. First, we observe that because the
  list $P^{d}_u$ of colors to try was drawn independently of all other
  lists, and has never been used before by the algorithm, it is necessarily
  independent of the rest of the algorithm state.

  A given color $(d, P^{d}_u[j])$ is only invalid if there exists some
  other vertex $w$ which has precisely this color. If this were the case, then
  the set \used would contain that color, because \used contains all
  colors on vertices $w$ with $\degr(w) = d$ and whose list of potential
  colors $P^{d}_{w}$ overlaps with $P^{d}_{u}$. Thus, the algorithm will
  detect any invalid colors in \Cref{step:check-color-validity}.

  The probability that the algorithm fails to find a valid color is:
  \begin{align*}
    \Pr[P^{d}_{u} \subseteq \used] 
    = \prod_{j=1}^{4 \log n} \Pr[P^{d}_{u}[j] \in \used] 
    = \prod_{j=1}^{4 \log n} \frac{|\used|}{2 d^2} 
    \le \frac{1}{2^{4 \log n}} 
    = \frac{1}{n^4} \,,
  \end{align*}
  where the inequality uses the fact that $|\used| \le \degr(u) = d$.
\end{proof}

Taking a union bound over the at most $n L / 2$ endpoints modified, we find that the total probability of
a recoloring failure in the algorithm is, by \Cref{lemma:recolor-fail-probability}, at most $(1/n^4) \cdot n L/2 \le 1/n^2$.

The rest of this section is dedicated to analyzing the space cost of \Cref{alg:ub-delta3}. In general, an adaptive adversary could try to construct a bad sequence of updates that causes the algorithm to store too many edges. The next two lemmas argue that for \Cref{alg:ub-delta3}, the adversary is unable to use his adaptivity for this purpose: he can do no worse than the worst oblivious adversary. Subsequently, \Cref{lem:delta3-space} shows that \Cref{alg:ub-delta3} does well in terms of space cost against an oblivious adversary, which completes the analysis.

\begin{lemma}\label{lem:delta3-condindep}
Let $\tau = (e_1,\chi_1,e_2,\chi_2,\ldots,\chi_{i-1},e_i)$ be the transcript of the edges $(e_1,\ldots,e_i)$ that an adversary provides to an implementation of \Cref{alg:ub-delta3}, and of the colorings $(\chi_1,\ldots,\chi_{i-1})$ produced by querying after each of the first $(i-1)$ edges was added. Let $\sigma = (e_{i+1},\ldots,e_j)$ be an arbitrary sequence of edges not in $\bigcup_{h=1}^{i} e_h$, and let $\gamma$ be a subsequence of $\sigma$. Conditioned on $\tau$, the next coloring $\chi_i$ returned is independent of the event that when the next edges in the input stream are $\sigma$, the algorithm will store $\gamma$ in its list $A$.
\end{lemma}

\begin{proof}
Let $G = \bigcup_{j=1}^{i} e_j$ be the graph containing all edges up to $e_i$,
and let $e_i = \{u,v\}$, so that $u$ is the
vertex recolored in Lines \ref{step:collect-neighbors} through
\ref{step:stop-on-first-color}. Let $\deg_G(x)$ be the degree of vertex $x$ in $G$. 
We can partition the array $[n] \times [L]$ of indices for random color lists $(P_x^{i})_{(x,i) \in [n] \times [L]}$ used by \Cref{alg:ub-delta3} into three groups, defined as follows:
\begin{align*}
  \cQ_> &= \{ (x,i) \in [n] \times [L] : i \ge \deg_G(x) + 1 \} \\ 
  \cQ_= &= \{ (u,\deg_G(u)) \} \\
  \cQ_< &= \{ (x,i) \in [n] \times [L] : i \le \deg_G(x) \} \setm \cQ_= \,.
\end{align*}
The next coloring $\chi_i$ returned by the algorithm depends only on $u$,
$G$, $\chi_{i-1}$, and the random color list $P^{\deg_G(u)}_u$. On the other
hand, the past colorings $(\chi_1,\ldots,\chi_{i-1})$ returned by the algorithm
depend only on $(e_1,\ldots,e_{i-1})$ and the color lists indexed by $\cQ_<$.
Finally, whether an edge $\{a,b\}$ is stored in the set $A$ in the future depends only
on the edges added up to that time and some of the color lists from $\cQ_>$, 
because (per Lines \ref{step:min-conf-level} to \ref{step:add-edge}) only color
lists $P_a^{i}$ and $P_b^{i}$ with $i \ge \max(\degr(a),\degr(b))$ are considered.
(Note that at the time the new edge is processed,
$\degr(a)$ and $\degr(b)$ will both be larger than $\deg_G(a)$ and $\deg_G(b)$ because 
Line~\ref{step:inc-deg-u} will have increased the vertex
degrees.) Also observe that the edges $(e_1,\ldots,e_{i})$ depend only on the colorings
$(\chi_1,\ldots,\chi_{i-1})$ and the randomness of the function $f$; thus the
transcript $\tau$ so far depends on the color lists in $\cQ_<$, but is independent
of the color lists in $\cQ_= \cup \cQ_>$. It follows that conditioned on the 
transcript $\tau$, the value $\chi_i$ of the next coloring returned is independent
of whether or not a given subset $\gamma$ of some future list $\sigma$ of edges
inserted is stored in the set $A$.
\end{proof}


\begin{lemma}\label{lem:delta3-adversary}
Let $m$ be an integer, and let $\eta$ be an adversary for \Cref{alg:ub-delta3}
for which the first $m$ edges submitted are always valid inputs for
\Cref{alg:ub-delta3}. (In other words, no edge is repeated, and no vertex 
attains degree $> L$.) Let $E$ be an event which depends only on the list of edges provided by $\eta$ and the subset of those edges which \Cref{alg:ub-delta3} stores
in the set $A$. Then there is a specific fixed input stream of length $m$ 
on which $\Pr[E]$ is at least as large as when $\eta$ chooses the inputs.\footnote{
In fact, one can prove that there is a distribution over fixed input streams so that the probability of $E$ occurring is exactly the same as when $\eta$ is used to pick the input.}
\end{lemma}

\begin{proof}
Let $\next$ be the function used by $\eta$ to pick the next input based on the
list of colorings produced so far, as per \Cref{sec:prelims}. We say
that a partial sequence of colorings $\rho = (\chi_1,\ldots,\chi_i)$ is \emph{pivotal}
for $\next$ if there exist two suffixes of $\rho$ given by $\pi = (\chi_1,\ldots,\chi_i,\chi_{i+1},\chi_{i+2},\ldots,\chi_{j})$ and $\pi' = (\chi_1,\ldots,\chi_i,\chi_{i+1}',\chi_{i+2}',\ldots,\chi_{j}')$, which first differ at coordinate
$i+1$, and where $\next(\pi) \ne \next(\pi')$.

If no sequence of colorings is pivotal for $\next$, then the adversary
only ever submits one stream of $m$ edges, and we are done. Otherwise, let $\rho$
be a maximal pivotal coloring sequence for $\next$, so that there does not exist
a coloring sequence $\pi$ which has $\rho$ as a prefix and which is also pivotal
for $\next$. We will construct a modified adversary $\tilde{\eta}$ given by $\widetilde{\next}$
which behaves the same on all coloring sequences that are not extensions of
$\rho$, which has at least the same probability of the event $E$, and where
neither $\rho$ nor any of its extensions is pivotal for $\widetilde{\next}$. If $\widetilde{\next}$ has no pivotal sequence of colorings, we are done; if not, we can repeat this process of finding modified adversaries with fewer pivotal sequences until that is the case.

Let $X = (X_1,\ldots,X_m)$ be the random variable whose $i$th coordinate corresponds
to the $i$th coloring returned by the algorithm, when the adversary is given by
$\next$. Write $X_{1..i}= (X_1,\ldots,X_i)$. Let $\rho = (\chi_1,\ldots,\chi_i)$. Because $\rho$ is a maximal pivotal
coloring sequence for $\next$, the next coloring returned---$X_{i+1}$---will determine the remaining
$(m - i - 1)$ edges sent by the adversary. Let $F$ be the random variable whose value
is this list of edges. For each possible value $\sigma$ of the conditional random
variable $(X_{i+1}| X_{1..i} = \rho)$, let $F_\sigma$ be the list of edges
sent when $(X_{1..i},X_{i+1}) = (\rho,\sigma)$. By \Cref{lem:delta3-condindep},
conditioned on the event $X_{1..i} = \rho$, and on the edges $F_\sigma$ being sent next,
$X_{i+1}$ and the event $E$ are independent. Thus
\begin{align*}
    \Pr[E \mid X_{1..i} = \rho] &= 
        \EE_{\sigma \sim X_{i+1} \mid X_{1..i} = \rho} \Pr[E \mid X_{1..i} = \rho, X_{i+1} = \sigma, F = F_\sigma] \\
        &= \EE_{\sigma \sim X_{i+1} \mid X_{1..i} = \rho} \Pr[E \mid X_{1..i} = \rho, F = F_\sigma] \,.
\end{align*}
Consequently, there is a value $\tilde\sigma$ where $\Pr[E \mid X_{1..i} = \rho, F = F_{\tilde\sigma}] \ge \Pr[E \mid X_{1..i} = \rho]$. We define $\widetilde{\next}$ so as to agree with $\next$, except that
after the coloring sequence $\rho$ has been received, the adversary now picks edges according to the sequence $F_{\tilde\sigma}$ instead of making a choice based on $X_{i+1}$. This change
does not reduce the probability of $E$ (and may even increase it conditioned on $X_{1..i} = \rho$). Finally, note 
that neither $\rho$ nor any extension thereof is pivotal for the function $\widetilde{\next}$ used by adversary $\tilde{\eta}$.
\end{proof}

\begin{lemma}[Bounding the space usage] \label{lem:delta3-space}
In the oblivious adversary setting, if a fixed stream of a graph $G$ with maximum degree $\Delta$ is provided to \Cref{alg:ub-delta3}, the total number of edges
stored by \Cref{alg:ub-delta3} is $O(n (\log n)^3 )$, with high probability.
\end{lemma}

\begin{proof}
We prove this by showing the maximum number of edges adjacent to any given 
vertex $v$ is $O( (\log n)^3 )$ with high probability. Let $d = \deg_G(v)$,
and $w_1,\ldots,w_d$ be the neighbors of $v$ in $G$, ordered by the order
in which the edges $\{v,w_i\}$ occur in the stream. For any $x \in [n]$,
write $P_x$ to be the random variable consisting of all of $x$'s color 
lists, $P_x := (P_{x}^{i})_{i\in[L]}$. Then for $i \in [d]$, define the indicator random
variable $Y_i$ to be $1$ iff the algorithm records edge $\{v,w_i\}$; since
$Y_i$ is determined by $P_v$ and $P_{w_i}$, the random variables $(Y_i)_{i \in [d]}$
 are conditionally independent given $P_v$.
 
Now, for each $i \in [d]$,
\begin{align*}
  \Pr[Y_i = 1 \mid P_v] 
    &= \Pr\left[ \bigvee_{j=i}^{L} \left\{ P_{w_i}^{j} \cap P_{v}^{j} \ne
    \varnothing \right\} ~\Big|~ P_v\right] \\
    & \le \sum_{j=i}^{L} \Pr\left[P_{w_i}^{j} \cap P_{v}^{j} \ne \varnothing \mid P_v\right] 
      = \sum_{j=i}^{L} \Pr\left[\exists h \in [4 \log n]:\, P_{w_i}^{j}[h] \in P_{v}^{j} \mid P_v\right] \\
    & \le \sum_{j=i}^{L} \sum_{h=1}^{4 \log n} \Pr\left[P_{w_i}^{j}[h] \in P_{v}^{j} \mid P_v\right]
      = \sum_{j=i}^{L} 4 \log n \cdot \frac{|P_{v}^{j}|}{2 j^2} \\
    & \le 16 (\log n)^2 \sum_{j=i}^{\infty} \frac{1}{j (j+1)}
      = \frac{16 (\log n)^2}{i} \,.
\end{align*}
Since $\EE[Y_i \mid P_v] = \Pr[Y_i = 1 \mid P_v]$, this upper bound gives
\begin{align*}
  \EE\left[\sum_{i=1}^{d} Y_i ~\Big|~ P_v\right] \le \sum_{i=1}^{d} \frac{16 (\log n)^2}{i} \le 32 (\log n)^3 \, ,
\end{align*}
using the fact that $\sum_{i=1}^d 1/i\leq \max\{2 \log d, 1\} \leq 2 \log n$.
Applying a form of the Chernoff bound:
\begin{align*}
\Pr\left[\sum_{i=1}^{d} Y_i \ge 2 \cdot 32 (\log n)^3 ~\Big|~ P_v \right] \le 
  \exp\left( -\frac{1}{3} \cdot 32 (\log n)^3 \right) \le \frac{1}{n^3} \,,
\end{align*}
which proves that the number of edges adjacent to $v$ is $\le 64 (\log n)^3$
with high probability, for any value of $P_v$.

Applying a union bound over all $n$ vertices, the probability that the maximum 
degree of the stored graph $A$ exceeds $64 (\log n)^3$ is less than $1/n^2$.
\end{proof}

Combining \Cref{lemma:recolor-fail-probability}, \Cref{lem:delta3-adversary} and \Cref{lem:delta3-space}, we arrive at the main result of this section.

\begin{theorem} \label{thm:ub-delta3}
  \Cref{alg:ub-delta3} is an adversarially robust $O(\Delta^3)$-coloring algorithm
  for insertion streams which stores $O(n (\log n)^4)$ bits related to the graph, 
  requires access to $\tO(nL)$ random bits, and even against an adaptive adversary succeeds with probability 
  $\ge 1 - O(1/n^2)$. \qed
\end{theorem}

A weakness of \Cref{alg:ub-delta3} is that it requires the algorithm be able
to access all $\tO(n L)$ random bits in advance. If we assume that the adversary is
limited in some fashion, then it may be possible to store $\le \tO(n)$ true
random bits, and use a pseudorandom number generator to produce the $\tO(n L)$
bits that the algorithm uses, on demand. For example, if the adversary only
can use $O(n / \log n)$ bits of space, using Nisan's PRG \cite{Nisan90} on
$\Omega(n)$ true random bits will fool the adversary. Alternatively, assuming
one-way functions exist, there is a classic construction~\cite{HastadILL99} to produce 
a pseudorandom number generator using $O(n)$ true random bits, which in polynomial
time generates $\poly(n)$ pseudorandom bits that any adversary limited to using polynomial time cannot distinguish with non-negligible probability from truly random bits.


%% file: ub-deltasquare.tex

\subsection{Sketch-Switching Based Algorithms for Turnstile Streams}\label{sec:ub-sketchswitch}

We present a class of sketch switching based algorithms for poly$(\Delta)$-coloring. First, we give an outline of a simple algorithm for insert-only streams that colors the graph using $O(\Delta^2)$ colors and $\tO(n\sqrt{\Delta})$ space, where $\Delta$ is the max-degree of the graph at the time of query. Next, we show how to modify it to handle deletions. This is given by \Cref{alg:deltasq-turnstile}, whose correctness is proven in \Cref{lem:turnstile-ub-square}. Then we describe how it can be generalized to get an $O(\Delta^k)$-coloring in $\tO(n\Delta^{1/k})$ space for insert-only streams for any constant $k \in \NN$. Finally, we prove the fully general result giving an $O(\Delta^k)$-coloring in $\tO(n^{1-1/k}m^{1/k})$ space for turnstile streams, which is given by \Cref{thm:turnstile-ub-general}. Finally, we discuss how we can get rid of some reasonable assumptions that we make for our algorithms and how to improve the query time.

Throughout this section, we make the standard assumption that the stream length $m$ for turnstile streams is bounded by poly$(n)$. When we say that a statement holds with high probability (w.h.p.), we mean that it holds with probability at least $1-1/\text{poly}(n)$. In our algorithms, we often take the \emph{product} of colorings of multiple subgraphs of a graph $G$. We define this notion below and record its key property.

\begin{definition}[Product of Colorings]\label{def:productcol}
Let $G_1 = (V,E_1), \ldots, G_k = (V,E_k)$ be graphs on a common vertex set $V$. Given a coloring $\chi_i$ of $G_i$, for each $i\in [k]$, the \emph{product} of these colorings is defined to be a coloring where each vertex $v\in V$ is assigned the color $\langle \chi_1(v),\chi_2(v),\ldots,\chi_k(v)\rangle$.
\end{definition}

\begin{lemma}\label{lem:productcol}
Given a proper $c_i$-coloring $\chi_i$ of a graph $G_i=(V,E_i)$ for each $i\in [k]$, the product of the colorings $\chi_i$ is a proper $(\prod_{i=1}^kc_i)$-coloring of $\cup_{i=1}^kG_i:= (V,\cup_{i=1}^k E_i)$.
\end{lemma}
\begin{proof}
An edge in $\cup_{i=1}^kG_i$ comes from $G_{i^*}$ for some $i^*\in [k]$, and hence the colors of its endpoints in the product coloring differ in the $i^*$th coordinate. For $i\in [k]$, the $i$th coordinate can take $c_i$ different values and hence the color bound holds.
\end{proof}

\mypar{Insert-Only Streams and $\bm{O(\Delta^2)}$-Coloring}
Split the $O(n\Delta)$-length stream into $\sqrt{\Delta}$ chunks of size $O(n\sqrt{\Delta})$ each. Let $\cA$ be a standard (i.e., oblivious-adversary) semi-streaming algorithm for $O(\Delta)$-coloring a graph (by \Cref{fact:ackalgo} and \Cref{fact:bcgsmallupdate}, such algorithms exist). At the start of the stream, initialize $\sqrt{\Delta}$ parallel copies of $\cA$, called $A_1, \ldots, A_{\sqrt{\Delta}}$; these will be our ``parallel sketches.'' At any point of time, only a suffix of this list of parallel sketches will be active.

We use the sketch switching idea of \cite{BenEliezerJWY20} as follows. With each edge insertion, we update each of the active parallel sketches. Whenever we arrive at the end of a chunk, we say we have reached a ``checkpoint'' and query the least-numbered active sketch (this is guaranteed to be ``fresh'' in the sense that it has not been queried before) to produce a coloring of the entire graph until that point. By design, the randomness of the queried sketch is independent of the edges it has processed. Therefore, it returns a correct $O(\Delta)$-coloring of the graph until that point, w.h.p. Henceforth, we mark the just-queried sketch as inactive and never update it, but continue to update all higher-numbered sketches. Thus, each copy of $\cA$ actually processes a stream independent of its randomness and hence, works correctly while using $\tO(n)$ space. By a union bound over all sketches, w.h.p., all of them generate correct colorings at the respective checkpoints and simultaneously use $\tO(n)$ space each, i.e., $\tO(n\sqrt{\Delta})$ space in total. 

Conditioned on the above good event, we can always return an $O(\Delta^2)$-coloring as follows. We store (buffer) the most recent partial chunk explicitly, using our available $\tO(n\sqrt{\Delta})$ space. Now, when a query arrives, we can express the current graph $G$ as $G_1\cup G'$, where $G_1$ is the subgraph of $G$ until the last checkpoint and $G'$ is the subgraph in our buffer. Observe that we computed an $O(\Delta(G_1))$-coloring of $G_1$ at the last checkpoint. Further, we can deterministically compute a $(\Delta(G')+1)$-coloring of $G'$ since we explicitly store it. We output the product of the colorings (\Cref{def:productcol}) of $G_1$ and $G'$, which must be a proper $O(\Delta(G_1)\cdot\Delta(G'))=O(\Delta(G)^2)$-coloring of the graph $G$ (\Cref{lem:productcol}). 

\mypar{Extension to Handle Deletions}
The algorithm above doesn't immediately work for turnstile streams. The chunk currently being processed by the algorithm may contain an update that deletes an edge which was inserted before the start of the chunk, and hence, is not in store. Thus, we can no longer express the current graph as a union of the graphs $G_1$ and $G'$ as above. Overcoming this difficulty complicates the algorithm enough that it is useful to lay it out more formally as pseudocode (see \Cref{alg:deltasq-turnstile}). This new algorithm maintains an $O(\Delta^2)$-coloring, works even on turnstile streams, and uses $\tO(\sqrt{mn})$ space. Note that while the blackbox algorithm $\cA$ used in \Cref{alg:deltasq-turnstile} might be any generic $O(\Delta)$-coloring semi-streaming algorithm with error $1/m$, it can be, for instance, chosen to be the one given by \Cref{fact:ackalgo} or, for insert-only streams, the one in \Cref{fact:bcgsmallupdate}. The former gives a tight $(\Delta+1)$-coloring but possibly large query time, while the latter answers queries fast using possibly a few more colors, up to $\Delta(1+\eps).$\footnote{In practice, however, the latter uses significantly fewer colors for most graphs since it's a $\kappa(1+\eps)$-coloring algorithm and $\kappa\leq \Delta$ always, and, in fact, $\kappa \ll \Delta$ for real world graphs.\cite{BeraCG20}}  

\begin{algorithm}
  \caption{Adversarially robust $O(\Delta^2)$-coloring in $\tO(\sqrt{nm})$ space for turnstile streams
    \label{alg:deltasq-turnstile}}
  \begin{algorithmic}[1]
    \Statex \textbf{Input}: Stream of edge insertions/deletions of an $n$-vertex graph $G=(V,E)$; parameter $m$
    \Statex
    \Statex \textbf{Require:} Semi-streaming algorithm $\cA$ that works on turnstile graph streams and provides an $O(\Delta)$-coloring with error $\le 1/m$ against an oblivious adversary
    \Statex
    
    \Statex \ul{\textbf{Initialize}:}
    \Let{$s$}{$C\cdot \sqrt{m/n}\log n$ for some sufficiently large constant $C$}
    \Let{$A_1,\ldots,A_s$}{independent parallel initializations of $\cA$}
    \Let{$c$}{$0$} \Comment{index into list $(A_1,\ldots,A_s)$}
    \Let{\clr}{$n$-vector of vertex colors, initialized to all-$1$s} 
    \Comment{valid $O(\Delta)$-coloring until last checkpoint}
    \Let{\degr}{$n$-vector of vertex degrees, initialized to all-$0$s}
    \Let{$G'$}{$(V,\varnothing)$} \Comment{buffer to store current chunk}
     \Let{\chunksize}{$0$} \Comment{current buffer size}
     \Let{\checkpointmaxdeg}{$0$} \Comment{max-degree at last checkpoint}

    \Statex
    \Statex \ul{\textbf{Process}(operation \textsc{op}, edge $\{u,v\}$):}
      \Comment{\textsc{op} says whether to insert or delete}
    \For{$i$ from $c + 1$ to $s$}
        \State $A_i$\textbf{\,.\,Process}(\textsc{op}, $\{u,v\}$) \Comment{if this aborts, report FAIL}
    \EndFor
    \If { \textsc{op} $=$ ``insert''}
        \State increment $\degr(u), \degr(v)$
        \State add $\{u,v\}$ to $G'$
    \ElsIf { \textsc{op} $=$ ``delete''}
        \State decrement $\degr(u), \degr(v)$
        \If{$\{u,v\}\in G'$}:
        \Comment{else, negative edge; not stored}
            \State delete $\{u,v\}$ from $G'$
        \EndIf
    \EndIf
    \vskip0.1cm
    \Let{\chunksize}{$\chunksize +1$}
    \Let{$\Delta$}{$\max_{v \in [n]} \degr(v)$}
    
    \If{$\chunksize = \sqrt{nm}$}:
        \State \textbf{NewCheckpoint}(~) \Comment{fixed checkpoint encountered}
    \Let{\chunksize}{$0$}
    \EndIf
    \If{$\Delta < \checkpointmaxdeg/2$}:
        \State \textbf{NewCheckpoint}(~)
        \Comment{ad hoc checkpoint created}
    \EndIf
    
    \Statex
    \Statex \ul{\textbf{Query(~)}:}
    \Let{$\text{\clr}'$}{$(\Delta_{G'}+1)$-coloring of $G'$}
    \State \Return{$\left\langle (\text{\clr}(v), \text{\clr}'(v)) : v\in [n]\right\rangle$}
      \Comment{take the product of the two colorings}
    \Statex
    \Statex \ul{\textbf{NewCheckpoint}(~):}
    \Let{$c$}{$c+1$}\Comment{switch to next fresh sketch}
    \Let{\clr}{$A_c$\textbf{\,.\,Query}(~)} \Comment{if $A_c$ fails, report FAIL}
    \Let{$G'$}{$(V,\varnothing)$}
    \Let{\checkpointmaxdeg}{$\max_{v \in [n]} \degr(v)$}
     \end{algorithmic}
\end{algorithm}

Before proceeding to the analysis, let us set up some terminology. Recall from \Cref{sec:prelims} that we work with strict graph turnstile streams, so each deletion of an edge $e$ can be matched to a unique previous token that most recently inserted $e$. An edge deletion, where the corresponding insertion did not occur inside the same chunk, is called a \emph{negative edge}. Call a point in the stream a \emph{checkpoint} if we use a \emph{fresh} parallel copy of $\cA$, i.e., a copy $A_i$ that hasn't been queried before, to generate an $O(\Delta)$-coloring of the graph at that point. We define two types of checkpoints, namely \emph{fixed} and \emph{ad hoc}. We have a fixed checkpoint at the end of each chunk; this means that whenever the last update of a chunk arrives, we compute a coloring of the graph seen so far using a fresh copy of $A$. The ad hoc checkpoints are made on the fly inside a current chunk, precisely when a query appears and we see that the max-degree of the current  graph is less than half of what it was at the \emph{last} checkpoint (which might be fixed or ad hoc). We now analyze \Cref{alg:deltasq-turnstile} in the following lemma.

\begin{lemma}\label{lem:turnstile-ub-square}
For any strict graph turnstile stream of length at most $m$ for a graph $G$ given by an adaptive adversary, the following hold simultaneously, w.h.p.:
\begin{itemize}[itemsep=1pt]
    \item[(i)] \Cref{alg:deltasq-turnstile} outputs an $O(\Delta^2)$-coloring after each query, where $\Delta$ is the maximum degree of the graph at the time a query is made.
    \item[(ii)] \Cref{alg:deltasq-turnstile} uses $\tO(\sqrt{mn})$ bits of space. 
\end{itemize}
\end{lemma}
\begin{proof}
Notice that \Cref{alg:deltasq-turnstile} splits the stream into chunks of size $\sqrt{mn}$. It processes one chunk at a time by explicitly storing all updates in it except for the negative edges. Nevertheless, when a negative edge arrives, the chunk size increases and importantly, we do update the appropriate copies of $\cA$ with it. Buffer $G'$ maintains the graph induced by the updates stored from the current chunk. The counter $c$ maintains the number of (overall) checkpoints reached. Whenever we reach a checkpoint, we re-initialize $G'$ to $G_0$, defined as the empty graph on the vertex set $V$. For $c\geq 1$, let $G_c$ denote the graph induced by all updates until checkpoint~$c$.

Note that answers to all queries (if any) that are made following some update before checkpoint $c$ depends only on sketches $A_i$ for some $i<c$ (if any). Thus, the random string used by the sketch $A_c$ is independent of the graph $G_c$. Hence, by the correctness guarantees of algorithm $\cA$, the copy $A_c$ produces a valid $O(\Delta)$-coloring \clr of $G_c$ with probability at least $1-1/m$. Furthermore, observe that an edge update before checkpoint $c$ is dependent on only the outputs of the sketches $A_j$ for $j< c$. However, we insert such an update only to copies $A_i$ for $i\geq c$. Therefore, the random string of any sketch $A_i$ is independent of the graph edges it processes. Thus, by the space guarantees of algorithm $\cA$, a sketch $A_i$ uses $\tO(n)$ space with probability $1-1/m$. By a union bound over all $s=O(\sqrt{m/n}\log n)$ copies, with probability at least $1-1/\text{poly}(n)$, for all $c\in [s]$, the sketch $A_c$ produces a valid $O(\Delta)$-coloring of the graph $G_c$ and uses $\tO(n)$ space. Now, conditioning on this event, we prove that (i) and (ii) always hold. Hence, in general, they hold with probability at least $1-1/\text{poly}(n)$.

 Consider a query made at some point in the stream. Since we keep track of all the vertex degrees and save the max-degree at the last checkpoint, we can compare the max-degree $\Delta$ of the current graph $G$ with $\Delta(G_c)$, where $c$ is the last checkpoint (can be fixed or ad hoc). In case $\Delta<\Delta(G_c)/2$, we declare the current query point as an ad hoc checkpoint $c+1$, i.e., we use the next fresh sketch $A_{c+1}$ to compute an $O(\Delta)$-coloring of the current graph $G_{c+1}$. Since we encounter a checkpoint, we reset \clr to this coloring and $G'$ to $G_0$, implying that $\clr'$ is just a $1$-coloring of the empty graph. Thus, the product of \clr and $\clr'$ that is returned uses only $O(\Delta)$ colors and is a proper coloring of the graph $G_{c+1}$.

In the other case that $\Delta>\Delta(G_c)/2$, we output the coloring obtained by taking a product of the $O(\Delta(G_c))$-coloring \clr at the last checkpoint $c$ and a $(\Delta(G')+1)$-coloring $\clr'$ of the graph $G'$. Note that we can obtain the latter deterministically since we store $G'$ explicitly. 
Observe that the edge set of the graph $G$ is precisely $(E(G_c)\setm F)\cup E(G')$, where $F$ is the set of negative edges in the current chunk. Since the coloring we output is a proper coloring of $G_c\cup G'$ (\Cref{lem:productcol}), it must be a proper coloring of $G$ as well because edge deletions can't violate it. It remains to prove the color bound.  The number of colors we use is at most $O(\Delta(G_c)\cdot \Delta(G'))$. We have checked that $\Delta\geq \Delta(G_c)/2$. Again, observe that $\Delta(G')\leq \Delta$ since $G'$ is a subgraph of $G$. Therefore, the number of colors used it at most $O(2\Delta\cdot \Delta)=O(\Delta^2)$. 

To complete the proof that (i) holds, we need to ensure that before the stream ends, we don't introduce too many ad hoc checkpoints so as to run out of fresh sketches to invoke at the checkpoints. We declare a point as an ad hoc checkpoint only if the max-degree has fallen below half of what it was at the last checkpoint (fixed or ad hoc). Therefore, along the sequence of ad hoc checkpoints between two consecutive fixed checkpoints (i.e., inside a chunk), the max-degree decreases by a factor of at least $2$. Hence, there can be only $O(\log \Delta_{\text{max}})=O(\log n)$ ad hoc checkpoints inside a single chunk, where $\Delta_{\text{max}}$ is the maximum degree of a vertex over all intermediate graphs in the stream. We have $O(\sqrt{m/n})$ chunks and hence, $O(\sqrt{m/n})$ fixed checkpoints and at most $O(\sqrt{m/n}\log n)$ ad hoc checkpoints. Thus, the total number of checkpoints is at most $s=O(\sqrt{m/n}\log n)$ and it suffices to have that many sketches initialized at the start of the stream. 

To verify (ii), note that since each chunk has size $\sqrt{mn}$, we use at most $\tO(\sqrt{mn})$ bits of space to store $G'$. Also, each of the $s$ parallel sketches takes $\tO(n)$ space, implying that they collectively use $\tO(ns)=\tO(\sqrt{mn})$ space. Storing all the vertex degrees takes $\tO(n)$ space. Therefore, the total space usage is $\tO(\sqrt{mn})$ bits. 
\end{proof}

\mypar{Generalization to $\bm{O(\Delta^k)}$-Coloring in $\bm{\tO(n\Delta^{1/k})}$ Space for Insert-Only Streams}
We aim to generalize the above result by attaining a color-space tradeoff. Again, for insert-only streams, it is not hard to obtain such a generalization and we outline the algorithm for this setting first.
\Cref{alg:deltasq-turnstile} shows that we need to use roughly $\tO(nr)$ space if we split the stream into $r$ chunks since we use a fresh $\tO(n)$-space sketch at the end of each chunk. Thus, to reduce the space usage, we can split the stream into smaller number of chunks. However, that would make the size of each chunk larger than our target space bound. Hence, instead of storing it entirely and coloring it deterministically as before, we treat it as a smaller stream in itself and recursively color it using space smaller than its size. To be precise, suppose that for any $d$, we can color a stream of length $nd$ using $O(\Delta^\ell)$ colors and $\tO(nd^{1/\ell})$ space for some integer $\ell$ (this holds for $\ell=2$ by \Cref{lem:turnstile-ub-square}). Now, suppose we split an $nd$-length stream into $d^{1/(\ell+1)}$ chunks of size $nd^{\ell/(\ell+1)}$. We use a fresh sketch at each chunk end or checkpoint to compute an $O(\Delta)$-coloring of the graph seen so far. We can then recursively color the subgraph induced by each chunk using $O(\Delta^{\ell})$ colors and $\tO\left(n\left(d^{\ell/(\ell+1)}\right)^{1/\ell}\right)=\tO(nd^{1/(\ell+1)})$ space. As before, taking a product of this coloring with an $O(\Delta)$-coloring at the last checkpoint gives an $O(\Delta^{\ell+1})$-coloring (\Cref{lem:productcol}) of the current graph in $\tO(nd^{1/(\ell+1)})$ space. The additional space used by the parallel sketches for the  $d^{1/(\ell+1)}$ many chunks is also $\tO(nd^{1/(\ell+1)})$. Therefore, by induction, we can get an $O(\Delta^k)$-coloring in $\tO(nd^{1/k})=O(n\Delta^{1/k})$ space for any integer $k$. We capture this result in \Cref{cor:insertonly-ub-general} after proving the more general result for turnstile streams. 

\mypar{Fully General Algorithm for Turnstile Streams}
Handling edge deletions with the above algorithm is challenging because of the same reason as earlier: a chunk of the stream may not itself represent a subgraph as it can have negative edges.  Therefore, it is not clear that we can recurse on that chunk with a blackbox algorithm for a graph stream. A trick to handle deletions as in \Cref{alg:deltasq-turnstile} faces challenges due to the recursion depth. We shall have an $O(\Delta)$-coloring at a checkpoint at each level of recursion that we basically combine to obtain the final coloring. Previously, we checked whether the max-degree has decreased significantly since the last checkpoint and if so, declared it as an ad hoc checkpoint. This time, due to the presence of checkpoints at multiple recursion levels, if the $\Delta$-value is too high at even a single level, we need to have an ad hoc checkpoint, which might turn out to be too many. We show how to extend the earlier technique to overcome this challenge and obtain the general result for turnstile streams, which achieves an $O(\Delta^k)$-coloring in $\tO(n^{1-1/k}m^{1/k})$ space for an $m$-length stream. 
\begin{figure}[!htb]
    \centering
    \includegraphics[width=\textwidth]{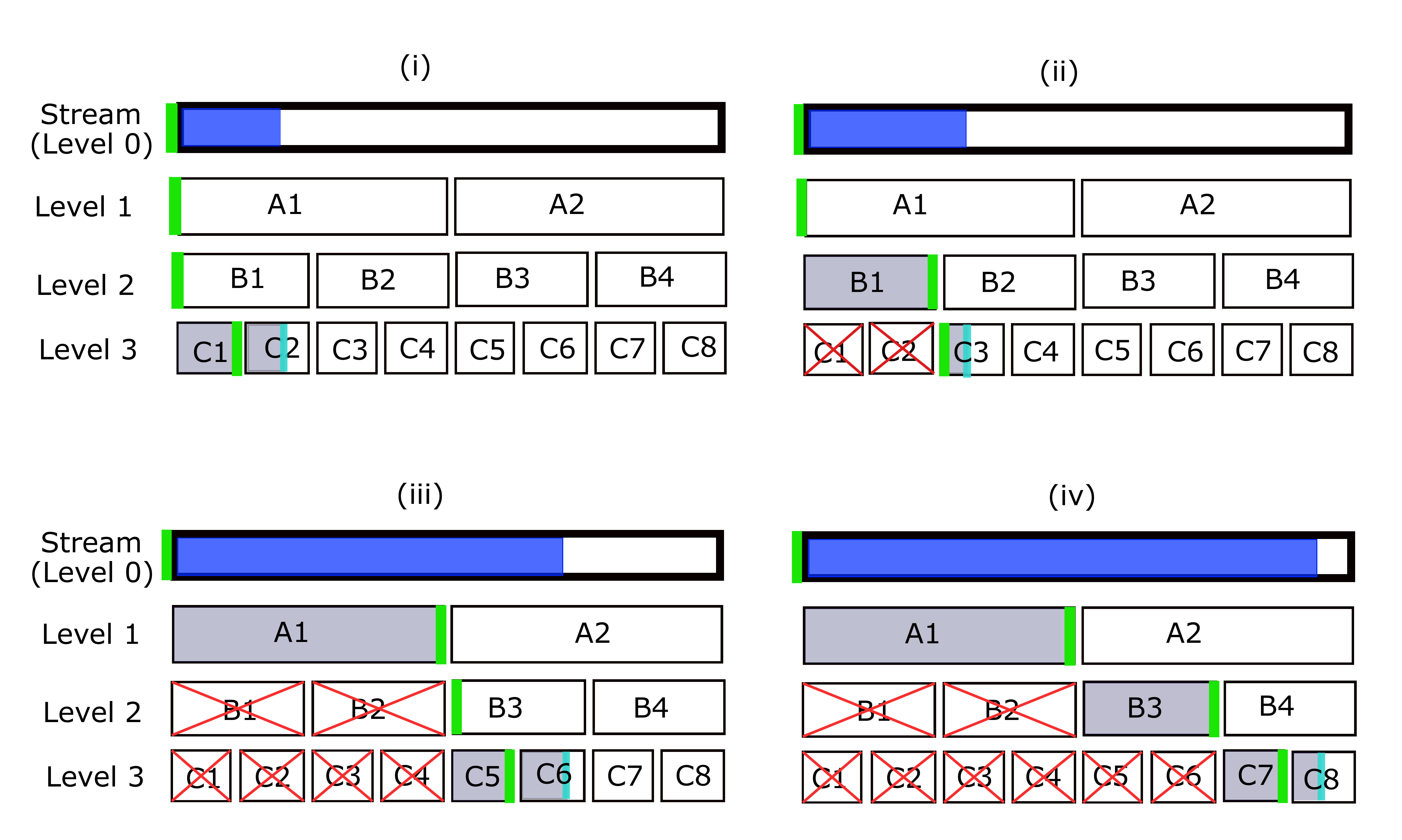}
    \caption{Certain states of the data structure of our $O(\Delta^k)$-coloring algorithm for $k=4$. We pretend that we always split into $d^{1/k}=2$ chunks. The stream is a level-$0$ chunk; $A_1, A_2$ are level-$1$ chunks; $B_1,\ldots,B_4$ are level-$2$; and $C_1,\ldots,C_8$ are level-$3$. For each state, the top blue bar shows the progress of the stream. Each level has a green vertical bar that represents the last checkpoint in that level. The chunks filled in gray represent the subgraphs defined as $G_i$. A partially filled chunk (endpoint colored cyan) is the current chunk from which the subgraph $G'$ is stored. A chunk is crossed out in red if it has been subsumed by a higher level chunk.    }
    \label{fig:multilevel}
\end{figure}
\begin{theorem} \label{thm:turnstile-ub-general}
For any strict graph turnstile stream of length at most $m$, and for any constant $k\in \NN$, there exists an adversarially robust algorithm $\cA$ such that the following hold simultaneously w.h.p.:
\begin{itemize}[itemsep=1pt]
    \item[(i)] After each query, $\cA$ outputs an $O(\Delta^k)$-coloring, where $\Delta$ is the max-degree of the current graph.
    \item[(ii)] $\cA$ uses $\tO(n^{1-1/k}m^{1/k})$ bits of space. 
\end{itemize} 
\end{theorem}

\begin{proof}
 The following framework is an extension of \Cref{alg:deltasq-turnstile} that would be given by the recursion idea discussed above. \Cref{fig:multilevel} shows the setup of our data structure. The full stream is the sole ``level-$0$'' chunk. Given $k$, we first split the edge stream into $d^{1/k}$ chunks of size $O(nd^{(k-1)/k})$ each, where $d=m/n$: these chunks are in ``level $1$.'' For $1\leq i\leq k-2$, recursively split each level-$i$ chunk into $d^{1/k}$ subchunks of size $O(nd^{(k-i-1)/k})$ each, which we say are in level $i+1$. Level $k-1$ thus has chunks of size $O(nd^{1/k})$. We explicitly store all updates in a level-$(k-1)$ chunk except the negative edges, one chunk at a time.
 
 Let $A$ be a turnstile streaming algorithm in the oblivious adversary setting that uses at most $\Delta(1+\eps)$ colors, where $\eps= 1/2k$, and $\tO(n)$ space, and fails with probability at most $1/(mn)$. By \Cref{fact:ackalgo}, such an algorithm exists.\footnote{By \Cref{fact:bcgsmallupdate}, another algorithm with these properties exists for insert-only streams.} At the start of the stream, for each $i\in [k-1]$, we initialize $s=O(d^{1/k}(k\log n)^k)$ parallel copies or ``level-$i$ sketches'' $A_{i,1},\ldots,A_{i,s}$ of $A$.  For each $i$, the level-$i$ sketches process the level-$i$ chunks. Henceforth, over the course of the stream, as soon as we reach the end of a level-$i$ chunk, since it subsumes all its subchunks, we re-initialize the level-$j$ sketches for each $j>i$. As before, at the end of each chunk in each level~$i$, we have a ``checkpoint'', i.e., we query a fresh level-$i$ sketch $A_{i,r}$ for some $r\in[s]$ to compute a coloring at such a point. Observe that this is a coloring of the subgraph starting from the last level-$(i-1)$ checkpoint through this point. Following previous terminology, we call these level-$i$ chunk ends as \emph{fixed} ``level-$i$ checkpoints''. (For instance, in \Cref{fig:multilevel}, in (i), the checkpoint at the end of chunk $C1$ is a fixed level-$3$ checkpoint, while in (iii), the checkpoint at the end of $A1$ is a fixed level-$1$ checkpoint.) 
 
 This time, we can also have what we call \emph{vacuous} checkpoints. The start of the stream is a vacuous level-$i$ checkpoint for each $0\leq i\leq k-1$. Further, for each $i\in [k-2]$, after the end of each level-$i$ chunk, i.e., immediately after a fixed level-$i$ checkpoint, we create a vacuous level-$j$ checkpoint for each $j>i$ (e.g., in \Cref{fig:multilevel}, in (i), the checkpoint at the start of $B1$ is a vacuous level-$2$ checkpoint, while in (ii), the one at the start of $C3$ is a vacuous level-$3$ checkpoint). It is, after all, a level-$j$ ``checkpoint'', so we want a coloring stored for the substream starting from the last level-$(j-1)$ checkpoint through this point. However, note, that for each $j>i$ this substream is empty (hence the term ``vacuous''). Hence, we don't waste a sketch for a vacuous checkpoint and directly store a $1$-coloring for that empty substream.
 
 We can also have \emph{ad hoc} level-$i$ checkpoints that we declare on the fly (when to be specified later). Just as we would do on reaching a fixed level-$i$ checkpoint, we do the following upon creating an ad hoc level-$i$ checkpoint: (i) query a fresh level-$i$ sketch to compute a coloring at this point (again, this is a coloring of the subgraph from the last level-$(i-1)$ checkpoint until this point), (ii) start splitting the remainder of the stream into subchunks of higher levels, (iii) 
 re-initialize the level-$j$-sketches for each $j>i$, and (iv) create vacuous level-$j$ checkpoints for each $j>i$. 
 
Any copy of algorithm $A$ that we use in any level is updated and queried as in \Cref{alg:deltasq-turnstile}: we update each copy as long as it is not used to answer a query of the adversary and whenever we query a sketch, we make sure that it has not been queried before. Therefore, as in \Cref{alg:deltasq-turnstile}, the random string of any copy is independent of the graph edges it processes. Hence, each sketch computes a coloring correctly and uses $\tO(n)$ space with probability at least $1-1/(mn)$. Taking a union bound over all $O(ds)=\tO(d^{1+1/k})$ sketches, we get that all of them simultaneously provide correct colorings and use $\tO(n)$ space each with probability at least $1-1/\text{poly}(n)$. Henceforth, as in the proof of \Cref{lem:turnstile-ub-square}, we condition on this event and show that (i) and (ii) always hold, thus proving that they hold w.h.p. in general.

 For $1\leq i \leq k-1$, define $G_i$ as the graph starting from the last level-$(i-1)$ checkpoint through the last level-$i$ checkpoint (in \Cref{fig:multilevel}, the last checkpoint in each level is denoted by a green bar, and the $G_i$'s are the graphs between two such consecutive bars; they are either filled with gray or empty; for instance, in (ii), $G_1=\emptyset$, $G_2=B1$, and $G_3=\emptyset$, while in (iv), $G_1=A1$, $G_2=B3$, and $G_3=C7$). Note that a graph $G_i$ might be empty: this happens when the last level-$i$ checkpoint is vacuous. Observe that we can express the current graph $G$ as $((G_1\cup G_2\cup\ldots\cup G_{k-1})\setm F) \cup G'$, where, $G'$ is the subgraph stored from the the current chunk in level $(k-1)$ (recall that it is induced by all updates in this chunk excluding the negative edges), and $F$ is the set of negative edges in the chunk. It is easy to see that we can keep track of the degrees so that we know $\Delta(G_i)$ for each $i$. We check whether there exists an $i\in [k-1]$ such that the max-degree $\Delta$ of the current graph $G$ is less than $\Delta(G_i)/(1+\eps)$. If not, we take the coloring from the last checkpoint of each level in $[k-1]$ and return the product of all these colorings with a $(\Delta(G')+1)$-coloring of $G'$ (\Cref{def:productcol}). We can compute the latter deterministically since we have $G'$ in store. Notice that the colorings at the checkpoints are valid colorings of $G_i$ for $i\in [k-1]$ using $1$ color if $G_i$ is empty and at most $(1+\eps)\Delta(G_i)\leq (1+\eps)^2\Delta$ colors otherwise. Also, $\Delta(G')\leq\Delta$ because $G'$ is a subgraph of $G$. Therefore, by \Cref{lem:productcol}, the total number of colors used to color $G$ is
\begin{align*}
    \prod_{i=1}^{k-1} (\max\{(1+\eps)^2\Delta,1\})\cdot (\Delta+1)\leq O\left((1+\eps)^{2k-2}\Delta^{k}\right)=O(\Delta^k) \,,
\end{align*}
since $2k-2< 2k= 1/\eps$. Finally, note that the product obtained will be a proper coloring of $G$ since the negative edges in $F$ cannot violate it. 

In the other case that there exists an $i$ such that $\Delta<\Delta(G_i)/(1+\eps)$, let $i^*$ be the first such $i$. We make this query point an ad hoc level-$i^*$ checkpoint. Also, the graph $G_{i^*}$ changes according to the definition above, and now the current graph $G$ is given by $G_1\cup\ldots\cup G_{i^*}$. Then, we return the product of colorings at the last checkpoints of levels $1,\ldots,i^*$. We know that these give $(1+\eps)\Delta(G_i)$-colorings for $i\in [i^*]$. Again, we have $\Delta(G_i)\leq \Delta$ since $G_{i}$ is a subgraph of $G$ for each $i$. Thus, the total number of colors used is
\begin{align*}
    \prod_{i=1}^{i^*}\left( (1+\eps)\Delta(G_i)\right) =(1+\eps)^{i^*}\Delta^{i^*}=O(\Delta^{k-1}) \,,
\end{align*}
since $i^*\leq k-1< 1/2\eps$. Therefore, in either case, we get an $O(\Delta^k)$-coloring. 

Now, as in the proof of \Cref{lem:turnstile-ub-square}, we need to prove that we have enough parallel sketches for the ad hoc checkpoints. Observe that we create an ad hoc level-$i$ checkpoint only when the current max-degree decreases by a factor of $(1+\eps)$ from the last checkpoint in level $i$ itself. Thus, along the sequence of ad hoc level-$i$ checkpoints between two consecutive non-ad-hoc (fixed or vacuous) level-$i$ checkpoints, the max-degree decreases by a factor of at least $(1+\eps)$. Therefore, there can be at most $\log_{1+\eps} n = O(\eps^{-1}\log n)=O(k\log n)$ such ad hoc checkpoints.

We show by induction that the number of ad hoc checkpoints in any level $i$ is $O(d^{1/k}(k\log n)^i)$. In level $1$, there is only $1$ vacuous checkpoint (at the beginning) and $d^{1/k}$ fixed checkpoints. Therefore, by the argument above, it can have $O(d^{1/k}(k\log n))$ ad hoc checkpoints; the base case holds. By induction hypothesis assume that it is true for all $i\leq j$. The number of vacuous checkpoints in level $j$ is equal to the number of fixed plus ad hoc checkpoints in levels $1,\ldots,j-1$. This is $\sum_{i=1}^{j-1} O(d^{1/k}(k\log n)^{i})=O(d^{1/k}k^j\log^{j-1} n)$ since $j<k$. The number of ad hoc checkpoints in level $j$ is $\log n$ times the number of vacuous plus fixed checkpoints in level $j$, which is $O(d^{1/k}k^j\log^{j-1} n\cdot \log n)=O(d^{1/k}(k\log n)^{j})$. Thus, by induction, there are $O(d^{1/k}(k\log n)^{i})$ ad hoc checkpoints in any level $i$. Therefore, the total number of checkpoints in level $i$ is also $O(d^{1/k}(k\log n)^{i}+d^{1/k}k^i\log^{i-1} n+d^{1/k})=O(d^{1/k}(k\log n)^{i})$. Thus, $s=O(d^{1/k}(k\log n)^{k})$ many parallel sketches suffice for each level. This completes the proof of~(i).

Finally, for (ii), as noted above, the $s$ parallel sketches of $A$ take up $\tO(n)$ space individually, and hence, $\tO(ns)=\tO(nd^{1/k})$ space in total (recall that $k=O(1)$. Additionally, the space usage to store the subgraph $G'$ from a level-$(k-1)$ chunk is $\tO(nd^{1/k})$. Hence, the total space used is $\tO(nd^{1/k})=\tO(n^{1-1/k}m^{1/k})$.
\end{proof}

The next corollary shows that the space bound for $O(\Delta^k)$-coloring on insert-only streams is $\tO(n\Delta^{1/k})$ and follows immediately from \Cref{thm:turnstile-ub-general} noting that $m=O(n\Delta)$ for such streams. Note that it works even for $k=\omega(1)$ since we don't have ad hoc checkpoints for insert-only streams and just $d^{1/k}$ sketches per level suffice.

\begin{corollary}\label{cor:insertonly-ub-general}
For any stream of edge insertions describing a graph $G$, and for any $k\in \NN$, there exists an adversarially robust algorithm $\cA$ such that the following hold simultaneously w.h.p.:
\begin{itemize}[topsep=5pt,itemsep=1pt]
    \item After each query, $\cA$ outputs an $O(\Delta^k)$-coloring, where $\Delta$ is the max-degree of the current graph.
    \item $\cA$ uses $\tO(n\Delta^{1/k})$ bits of space. \qed
\end{itemize} 
\end{corollary}

\mypar{Implementation Details: Update and Query Time}
Observe that if we use the algorithm by \cite{AssadiCK19} or \cite{BeraCG20} as a blackbox, then, to answer each query of the adversary, the time we spend is the post-processing time of these algorithms, which are  $\tO(n\sqrt{\Delta})$ and $\tO(n)$ respectively. Although in the streaming setting, we don't care that much about the time complexity, such a query time might be infeasible in practice since we can potentially have a query at every point in the stream. Thus, ideally, we want an algorithm that \emph{maintains} a coloring at every point in the stream spending a reasonably small time to update the solution after each edge insertion/deletion. This is similar to the dynamic graph algorithms setting, except here, we are asking for more: we want to optimize the space usage as well.

The algorithm by \cite{BeraCG20} broadly works as follows for insert-only streams. It partitions the vertex set into a number of clusters and stores only intra-cluster edges during stream processing. In the post-processing phase, it colors each cluster using an offline $(\Delta+1)$-coloring algorithm with pairwise disjoint palettes for the different clusters. This attains a desired $(1+\eps)\Delta$-coloring of the entire graph. We observe that instead, we can color each cluster on the fly using a dynamic $(\Delta+1)$-coloring algorithm such as the one by \cite{HenzingerP20} that takes $O(1)$ amortized update time for maintaining a coloring. A stream update causes an edge insertion in at most one cluster and hence, the update time is the same as that required for a single run of \cite{HenzingerP20}. The  \cite{BeraCG20} algorithm runs roughly $O(\log n)$ parallel sketches, and hence, we can maintain a $(1+\eps)\Delta$-coloring of the graph in $\tO(1)$ update time while using the same space as \cite{BeraCG20}, which is $\tO(\eps^{-2}n)$. This proves \Cref{fact:bcgsmallupdate}. 

If we use this algorithm as the blackbox algorithm $A$ in our adversarially robust algorithm for $O(\Delta^k)$-coloring in insert-only streams, we get $\tO(1)$ amortized update time for each parallel copy of $A$, implying an $\tO(s)$ amortized update time in total, where $s$ is the number of parallel sketches used. We, however, also need to process a buffer deterministically, where we cannot use the aforementioned algorithm since it's randomized. We can use the deterministic dynamic $(\Delta+1)$-coloring algorithm by \cite{BhattacharyaCHN18} for this part to get an additional $\tO(1)$ amortized update time. Thus, overall, our update time is $\tO(s)=\tO((m/n)^{1/k})=\tO(\Delta^{1/k})$. Finally, we can think of the algorithm as maintaining an $n$-length vector representing the coloring and making changes to its entries with every update while spending $\tO(\Delta^{1/k})$ time in the amortized sense. Hence, there's no additional time required to answer queries. This is a significant improvement over a query time of $\tO(n\sqrt{\Delta})$ or $\tO(n)$. 

\mypar{Removing the Assumption of Prior Knowledge of $\bm{m}$}
Observe that in \Cref{alg:deltasq-turnstile} as well as the algorithm described in \Cref{thm:turnstile-ub-general}, we assume that a value $m$, an upper bound on the number of edges, is given to us in advance. Without it, we do not know how many sketches to initialize at the start of the stream. A typical guessing trick does not seem to work since even the last sketch needs to process the entire graph and cannot be started ``on the fly'' if we follow our framework. In this context, we note the following. First, knowledge of an upper bound on the number of edges is a reasonable assumption, especially for turnstile streams, since an algorithm typically knows how large of an input stream it can handle. Second, for insert-only streams, we can always set $m=n\Delta/2$ if an upper bound $\Delta$ on the max-degree of the final graph is known; a knowledge of such a bound is reasonable since $f(\Delta)$-coloring is usually studied with a focus on bounded-degree graphs. Third, we can remove the assumption of knowing either $m$ or $\Delta$ for insert-only streams at the cost of a factor of $\Delta$ in the number of colors and an additive $\tO(n)$ factor in space, which we outline next. 

At the beginning of the stream, we initalize $\floor{\log n}$ copies of the oblivious $O(\Delta)$-coloring semi-streaming algorithm $A$ for the checkpoints where $\Delta$ first attains values of the form $2^i$ for some $i\in [\floor{\log n}]$. For each $i$, the substream between the checkpoints with $\Delta =2^i$ and $\Delta=2^{i+1}$ can be handled using our algorithm as a blackbox since we know that the stream length is at most $2^{i+1}n$. This way, we need not initialize $O(D^{1/k})$ sketches for $D \gg \Delta_{\text{max}}$ at the very beginning of the stream, where $\Delta_{\text{max}}$ is the final max-degree of the graph, and incur such a huge factor in space; we can initialize the $d^{1/k}$ sketches for the substream with $d\le\Delta\le 2d$ only when (if at all) $\Delta$ reaches the value $d$. Thus, the maximum space used is $O(n\Delta_{\text{max}}^{1/k})$, which we can afford. When queried in a substream between checkpoints at $\Delta =2^i$ and $\Delta=2^{i+1}$, we use our $O(\Delta^k)$-coloring algorithm to get a coloring of the substream, and take product with the $O(\Delta)$-coloring at the checkpoint at $\Delta=2^i$. Thus, we get an $O(\Delta^{k+1})$-coloring of the current graph. The additional space usage is $\tO(n)$ due to the initial $\floor{\log n}$ sketches taking up $\tO(n)$ space each; hence, the total space usage is still $O(n\Delta_{\text{max}}^{1/k})$.